\documentclass{lmcs} 
\pdfoutput=1
\usepackage[utf8]{inputenc}

\usepackage{lastpage}
\lmcsdoi{20}{4}{21}
\lmcsheading{}{\pageref{LastPage}}{}{}%
{Jun.~06,~2023}{Dec.~04,~2024}{}


\theoremstyle{plain} 

\usepackage{lineno,hyperref}
\usepackage{xargs}
\usepackage{amsthm}
\usepackage{amsmath}
\usepackage{oxford}

\usepackage{algorithm}
\usepackage[noend]{algpseudocode}
\usepackage{varwidth}

\newcommand{\ignore}[1]{}

\newcommand{\SetN}{\mathbb{N}}
\newcommand{\SetZ}{\mathbb{Z}}
\newcommand{\SetR}{\mathbb{R}}

\newcommand{\set}[2]{\ensuremath{\argint{\{}{\argext{#1}{\allowbreak:\allowbreak}{#2}}{\}}}}

\newcommand{\card}[1]{\mthempty{\argint{\vert}{#1}{\vert}}}

\newcommand{\LTL}{\mthfun{LTL}\xspace}
\newcommand{\GRone}{\mthfun{GR(1)}\xspace}
\def\TEMPORAL#1{\mbox{\small\boldmath$\mathbf{#1}$}}
\def\ltlnext{\TEMPORAL{X}}
\def\sometime{\TEMPORAL{F}} 
\def\always{\TEMPORAL{G}}
\def\until{\,\TEMPORAL{U}\,}
\def\Nat{\mathbb{N}}

\newcommand{\Ag}{\mthset{N}}
\newcommand{\Ac}{\mthset{Ac}}
\newcommand{\AcProf}{\vec{\Ac}}
\newcommand{\St}{\mthset{St}}
\newcommand{\AP}{\mthset{AP}}

\renewcommand{\Game}{\mthname{G}}

\newcommand{\pun}{\mthsym{pun}}

\newcommand{\labFun}{\lambda}

\newcommand{\trnFun}{\mthfun{tr}}
\newcommand{\act}{\mthsym{a}}
\newcommand{\jact}{\vec{\mthsym{a}}}

\newcommand{\StrSet}{\mthset{Str}}

\newcommand{\strElm}{\sigma}

\newcommand{\strpElm}{\mthelm{\vec{\sigma}}}

\newcommand{\NE}{\mthset{NE}}

\newcommand{\winsym}{\mthset{Win}}
\newcommandx{\Win}[3][1=, 2=, 3=]
{\mthset{\winsym#3}[#1][#2]}

\newcommand{\presym}{\mthfun{Pre}}
\newcommandx{\Pre}[3][1=, 2=, 3=]
{\mthset{\presym#3}[#1][#2]}

\newcommand{\eqsym}{\mthfun{Eq}}
\newcommandx{\Eq}[3][1=, 2=, 3=]
{\mthset{\eqsym#3}[#1][#2]}

\newcommand{\WI}{\mthset{WI}}

\newcommand{\SI}{\mthset{SI}}

\newcommand{\K}{\mathcal{K}}

\newcommandx{\AFW}[5][1=, 2=, 3=, 4=, 5=]
{\txtargname{AFW#5{\small\argint{$[$}{#1}{$]$}}}[#2][#3]{#4}\xspace}

\newcommand{\MP}[1][]{%
	\ifthenelse{\equal{#1}{}}{{{\sf mp}}}{{{\sf mp}}(#1)}%
\xspace}

\providecommand{\strFun}[1][]{\mthfun{\sigma}}
\providecommand{\pstrFun}[1][]{\mthfun{\sSym}}

\newcommand{\weak}{{\normalfont\textsc{Weak Implementation}}\xspace}
\newcommand{\weakc}{{\normalfont\textsc{Weak Implementation Complement}}\xspace}
\newcommand{\strong}{{\normalfont \textsc{Strong Implementation}}\xspace}
\newcommand{\strongc}{{\normalfont \textsc{Strong Implementation Complement}}\xspace}
\newcommand{\optwi}{{\normalfont \textsc{Opt-WI}}\xspace}
\newcommand{\optsi}{{\normalfont \textsc{Opt-SI}}\xspace}
\newcommand{\uoptwi}{{\normalfont \textsc{UOpt-WI}}\xspace}
\newcommand{\uoptsi}{{\normalfont \textsc{UOpt-SI}}\xspace}
\newcommand{\exactwi}{{\normalfont \textsc{Exact-WI}}\xspace}
\newcommand{\exactsi}{{\normalfont \textsc{Exact-SI}}\xspace}
\newcommand{\utwi}{{\normalfont \textsc{UT-WI}}\xspace}
\newcommand{\utsi}{{\normalfont \textsc{UT-SI}}\xspace}

\newcommand{\etwi}{{\normalfont \textsc{ET-WI}}\xspace}
\newcommand{\etsi}{{\normalfont \textsc{ET-SI}}\xspace}
\newcommand{\QSAT}{{\normalfont \textsc{QSAT}\xspace}}
\newcommand{\MQSAT}{{\normalfont \textsc{MinQSAT}\xspace}}
\newcommand{\MWQSAT}{{\normalfont \textsc{Weighted MinQSAT}\xspace}}

\newcommand{\wFun}{\mthfun{w}}
\newcommand{\cFun}{\mthfun{c}}
\newcommand{\cost}{\mthfun{cost}}

\newcommand{\LTLlim}{\mthfun{LTL^{Lim}}\xspace}

\def\avg{{\sf avg}}

\def\src{{\sf src}}
\def\trg{{\sf trg}}
\def\IN{{\sf in}}
\def\OUT{{\sf out}}

\newcommand{\pay}{\mthfun{pay}}

\newcommand{\LP}{\mthfun{LP}}

\def\pspace{\mthfun{PSPACE}\xspace}
\def\fpspace{\mthfun{FPSPACE}\xspace}

\def\np{\mthfun{NP}\xspace}
\def\conp{\text{co-}\mthfun{NP}\xspace}

\def\FP{\mthfun{FP}\xspace}

\def\usw{\mthfun{usw}\xspace}

\def\esw{\mthfun{esw}\xspace}

\def\mi{\mthfun{MinNE}\xspace}

\def\DP{\mthfun{D^P}\xspace}
\def\DPTwo{\mthfun{D^P_{\normalfont \textrm{2}}}\xspace}

\newcommand{\sink}{\mthsym{sink}}
\newcommand{\source}{\mthsym{source}}

\newcommand{\setx}{\mathbf{x}}
\newcommand{\sety}{\mathbf{y}}
\newcommand{\vecx}{\vec{\mathbf{x}}}
\newcommand{\vecy}{\vec{\mathbf{y}}}
\newcommand{\term}{\mathbf{t}}
\newcommand{\T}{\mathbf{T}}
\newcommand{\SigmaPTwo}{\Sigma^{\mthfun{P}}_2}
\newcommand{\DeltaPTwo}{\Delta^{\mthfun{P}}_2}
\newcommand{\DeltaPThree}{\Delta^{\mthfun{P}}_3}

\usepackage{xcolor}

\theoremstyle{definition}
\newtheorem{example}[thm]{Example}

\usepackage{xargs}

\usepackage{xspace}

\usepackage{xstring}

\usepackage{boolexpr}

\usepackage{amssymb}
\usepackage{amsfonts}
\usepackage{latexsym}
\usepackage{textcomp}
\usepackage{pifont}

\newcommand{\argemp}[2]{\if&#1&\else#2\fi}

\newcommand{\argdef}[2]{\if&#1&#2\else#1\fi}

\newcommand{\argint}[3]{\if&#2&\else#1#2#3\fi}

\newcommand{\argext}[3]{\if&#1&#3\else#1\if&#3&\else#2#3\fi\fi}

\newcommandx{\mthfnt}[3][1=, 2=0]{{
	\IfStrEqCase{#1}
	{%
		{}%
		{#3}%
		{Name}%
		{%
			\IfStrEqCase{#2}
			{%
				{0}{\mathcal{#3}}%
				{1}{\mathscr{#3}}%
				{2}{\mathfrak{#3}}%
				{3}{\mathbb{#3}}%
			}
			[\ensuremath{\clubsuit}]%
		}%
		{Set}%
		{%
			\IfStrEqCase{#2}
			{%
				{0}{\mathrm{#3}}%
				{1}{\mathsf{#3}}%
				{2}{\mathbb{#3}}%
				{3}{\mathbf{#3}}%
			}
			[\ensuremath{\clubsuit}]%
		}%
		{Fun}%
		{%
			\IfStrEqCase{#2}
			{%
				{0}{\mathsf{#3}}%
				{1}{\mathrm{#3}}%
			}
			[\ensuremath{\clubsuit}]%
		}%
		{Rel}%
		{%
			\IfStrEqCase{#2}
			{%
				{0}{\mathit{#3}}%
				{1}{\mathtt{#3}}%
			}
			[\ensuremath{\clubsuit}]%
		}%
		{Sym}%
		{%
			\IfStrEqCase{#2}
			{%
				{0}{\mathtt{#3}}%
				{1}{\mathbf{#3}}%
			}
			[\ensuremath{\clubsuit}]%
		}%
		{Elm}%
		{\mathnormal{#3}}
	}
[\ensuremath{\clubsuit}]%
}}

\newcommand{\mthsub}[1]{\argemp{#1}{\ensuremath{_{\mathnormal{#1}}}}}

\newcommand{\mthsup}[1]{\argemp{#1}{\ensuremath{^{\mathnormal{#1}}}}}

\newcommandx{\mth}[5][1=, 2=0, 4=, 5=]{{\ensuremath{\mthfnt[#1][#2]{#3}\mthsub{#4}\mthsup{#5}}}}

\newcommandx{\mtharg}[6][1=, 2=0, 4=, 5=]{{\mth[#1][#2]{#3}[#4][#5]\ensuremath{\argint{(}{#6}{)}}}}

\newcommand{\mthempty}{\mth[][]}

\newcommand{\mthstyname}{0}
\newcommand{\mthname}[1][]{\mth[Name][\argdef{#1}{\mthstyname}]}

\newcommand{\mthstyset}{0}
\newcommand{\mthset}[1][]{\mth[Set][\argdef{#1}{\mthstyset}]}

\newcommand{\mthstyfun}{0}
\newcommand{\mthfun}[1][]{\mth[Fun][\argdef{#1}{\mthstyfun}]}

\newcommand{\mthstysym}{0}
\newcommand{\mthsym}[1][]{\mth[Sym][\argdef{#1}{\mthstysym}]}

\newcommand{\mthstyelm}{0}
\newcommand{\mthelm}[1][]{\mth[Elm][\argdef{#1}{\mthstyelm}]}

\newcommand{\tuple}[1]
{\ensuremath{\!\argint{\langle}{#1}{\rangle}}}

\theoremstyle{plain}
\newtheorem{theorem}[thm]{Theorem}

\newtheorem{corollary}[thm]{Corollary}
\newtheorem{proposition}[thm]{Proposition}

\theoremstyle{definition}
\newtheorem{definition}[thm]{Definition}
\newtheorem{remark}[thm]{Remark}

\usepackage{tikz}
\usetikzlibrary{arrows,shapes,snakes}

\usepackage{wrapfig}

\newcommand{\GridExTwo}
{
	\begin{figure}
		\begin{center}
			\vspace{-15pt}
			\begin{tikzpicture}[every node/.style={minimum size=1cm-\pgflinewidth, outer sep=0pt}]
				\draw[step=1cm,color=black] (-1,-1) grid (2,1);
				
				\node[circle, thick, fill=white, minimum size=0.5cm,draw] at (-.5,.5){};
				\node[rectangle, fill=black, minimum size=0.5cm] at (1.5,.5){};
				\draw (-1,1) -- (-1.5,1.5) node[pos=0.75,xshift=0.25cm] {~$x$~} node[pos=0.75,yshift=-0.25cm] {$y$~};
				\node[] at (-.5,1.5) {0};
				\node[] at (.5,1.5) {1};
				\node[] at (1.5,1.5) {2};
				
				\node[] at (-1.5,0.5) {0};
				\node[] at (-1.5,-0.5) {1};
				
			\end{tikzpicture}		
			\vspace{-15pt}
		\end{center}
		\caption{Graphical representation from Example \ref{ex:grid1}.}
		\label{fig:GridExTwo}
	\end{figure}
	
}

\newcommand{\extsTwo}
{
	\begin{figure}[]
		\begin{center}
			\begin{tikzpicture}
				[->,>=latex,shorten >=1pt,auto,node distance=2cm, auto,main node/.style={rectangle,draw}]
				
				\node[main node, label={$  $}, label=below:{}] (1) [] {$ (0,0)_{0} $};
				
				\node[main node,  label={$ $}, label=below:{}] (4) [right of=1, yshift=1cm,xshift=1cm] {$ (1,1)_0 $};
				\node[main node,  label={}, label=below:{$ $}] (3) [below of=4] {$ (1,0)_1 $};
				\node[main node, label={$ $}, label=below:{}] (2) [above  of=4] {$ (1,0)_0  $};

				\node[main node,  label={$ $}, label=below:{}] (5) [below of=3] {$ (1,1)_1 $};
				\node[main node,  label={$ $}, label=below:{}] (6) [right of=1, xshift=4cm] {$ (2,1)_{1} $};
				
				\path
				++ (-1,0) edge node[]{} (1)
				(1) edge [in=180, bend left, <->] node{$ e_1 $} (2)
				(1) edge [in=180, bend right, <-] node[]{$ e_3 $} (3)
				(1) edge [out=270,in=180, bend left, <->] node[xshift=0.5cm,yshift=-0.5cm]{$ e_2 $} (4)
				(1) edge [out=270,in=180, bend right, <-] node[xshift=-0.5cm, yshift=-0.5cm]{$ e_4 $} (5)
				
				(6) edge [in=0, bend right, <-] node[yshift=0.5cm,xshift=0.3cm]{$ e_5 $} (2)
				(6) edge [in=0, bend left, <->] node[xshift=-0.4cm, yshift=0.4cm]{$ e_7 $} (3)
				(6) edge [in=0, bend right, <-] node[yshift=0cm, xshift=0.2cm]{$ e_6 $} (4)
				(6) edge [in=0, bend left, <->] node[xshift=-0.2cm]{$ e_8 $} (5)
				
				(2) edge [in=90, out=270, <->] node[xshift=0cm]{$ e_9 $} (4)
				
				(3) edge [in=90, out=270, <->] node[xshift=0cm]{$ e_{10} $} (5)
				;
				
			\end{tikzpicture}
		\end{center}
		\caption{Transition system representing the movements and payoff function for robot $ \bigcirc $.}
		\label{fig:ex1tsTwo}
	\end{figure}
}

\usepackage{booktabs}
\usepackage{ifthen}
\usepackage{cancel}

\usepackage{easyReview}
\begin{document}

\title[Designing Equilibria in Concurrent Games]{Designing Equilibria in Concurrent Games with Social Welfare and Temporal Logic Constraints\rsuper*}
\titlecomment{{\lsuper*}Preliminary version appeared in the \textit{Proceedings of 30th International Conference on Concurrency Theory (CONCUR 2019)}~\cite{GNPW19b}.}

\author[J.~Gutierrez]{Julian Gutierrez\lmcsorcid{0000-0002-1091-8232}}[a]	
\address{University of Sussex}	
\email{J.Gutierrez@sussex.ac.uk} 

\author[M.~Najib]{Muhammad Najib\lmcsorcid{0000-0002-6289-5124}}[b] 
\address{School of Mathematical and Computer Sciences, Heriot-Watt University}	%
\email{m.najib@hw.ac.uk}  

\author[G.~Perelli]{Giuseppe Perelli\lmcsorcid{0000-0002-8687-6323}}[c] 
\address{Department of Computer Science, Sapienza University of Rome}	
\email{perelli@di.uniroma1.it}  

\author[M.~Wooldridge]{Michael Wooldridge\lmcsorcid{0000-0002-9329-8410}}[d]	
\address{Department of Computer Science, University of Oxford}	
\email{mjw@cs.ox.ac.uk}  

\begin{abstract}
In game theory, {\em mechanism design\/} is concerned with the design of
incentives so that a desirable outcome will be achieved under the assumption that players act rationally. In this paper, we explore the concept of \textit{equilibrium design}, where incentives are designed to obtain a desirable equilibrium that satisfies a specific temporal logic property. Our study is based on a framework where system specifications are represented as temporal logic
formulae, games as quantitative concurrent game structures, and players' goals as mean-payoff objectives. We consider system specifications given by LTL and GR(1) formulae, and show that designing incentives to ensure that a given temporal logic property is satisfied on some/every Nash equilibrium of the game can be achieved in
PSPACE for LTL properties and in NP/$\SigmaPTwo$ for GR(1) specifications. We also examine the complexity of related decision and optimisation problems, such as optimality and uniqueness of solutions, as well as considering social welfare, and show that the
complexities of these problems lie within the polynomial hierarchy. Equilibrium design can be used as an alternative solution to
rational synthesis and verification problems for concurrent games with
mean-payoff objectives when no solution exists or as a technique to
repair concurrent games with undesirable Nash equilibria in an optimal way.
\end{abstract}

\maketitle


\section{Introduction}
Over the past decade, there has been increasing interest in the use of
game-theoretic equilibrium concepts such as Nash equilibrium in the
analysis of concurrent and multi-agent systems (see,
{\em e.g.},~\cite{AKP18,AminofMMR16,BouyerBMU15,FismanKL10,GHPW17,GutierrezHW17-aij,KupfermanPV16}). 
This work views a concurrent system as a game, with system components
(agents) corresponding to players in the game, which are assumed to be
acting rationally in pursuit of their individual
preferences. Preferences may be specified by associating with each
player a temporal logic goal formula, which the player desires to see
satisfied, or by assuming that players receive rewards in each state
the system visits, and seek to maximise the average reward they
receive (the \emph{mean-payoff}). A further possibility is to combine
goals and rewards: players primarily seek the satisfaction of their
goal, and only secondarily seek to maximise their mean-payoff. The key
decision problems in such settings relate to what temporal logic
properties hold on computations of the system that may be generated by
players choosing strategies that form a game-theoretic (e.g., Nash)
equilibrium. These problems are typically computationally complex,
since they subsume temporal logic synthesis~\cite{PnueliR89}. If
players have \LTL goals, for example, then checking whether an \LTL
formula holds on some Nash equilibrium path in a concurrent game is
2EXPTIME-complete~\cite{FismanKL10,GutierrezHW15,GutierrezHW17-aij}, 
rather than only PSPACE-complete as it is the case for model checking. This represents a major computational barrier for the practical analysis and automated verification of reactive, concurrent, and multi-agent systems modelled as multi-player games. 

A classic problem in game theory is that individually rational choices can result in outcomes that are highly undesirable, and concurrent games also
fall prey to this problem. To illustrate this, consider the scenario in the following example.

\GridExTwo

\begin{example}\label{ex:grid1}
	Consider a system with two robot agents operating in an environment modelled as a $ 3 \times 2 $ grid world. Initially, the robots are located at two corners, as shown in Figure \ref{ex:grid1}. Each robot can move one square horizontally, vertically, or diagonally (similar to the way a king moves in chess). Each move costs the robot 1 unit of energy and incurs a payment of $-1$. The task of robot $ \bigcirc $ (resp. $ \blacksquare $) is to visit square $ (2,1) $ (resp. $ (0,1) $), for example, to deliver parcels. To model this objective, we give the robots a payment of $ 3 $ when they reach their target squares. We assume that at each time-step, each robot must make a move and cannot remain in the same position for two consecutive time-steps. Furthermore, each robot wants to maximise the sum of payments it receives.

Suppose that having two robots occupying the same grid square (i.e., having the same coordinates) is considered undesirable from a global perspective because it increases the likelihood of collisions. To maximise their total payoffs, the robots will choose routes that minimise the number of steps needed to reach their target squares. Observe that the minimum number of steps required for the robots to reach their respective target squares is 2, and there are two routes that achieve this: via $ (1,0) $ and $ (1,1) $. However, if both robots choose the same shortest route (e.g., via $ (1,1) $), this results in an undesirable outcome. Moreover, since the shortest routes correspond to the best strategies for each robot, these outcomes are considered stable from a game-theoretic perspective.
\end{example}


These concerns have motivated the development of
techniques for modifying games, in order to avoid undesirable equilibria, or
to facilitate desirable equilibria. In game theory, the field of \emph{mechanism design} is concerned with designing a game such that, if players behave rationally, then a
desired outcome will be obtained~\cite{OR94}. Direct incentives, for example in the form of taxation or subsidy, are probably the most important tools used in mechanism design. 

This paper explores the design of incentive schemes for concurrent games so as to achieve a desired outcome, a concept we refer to as \textit{equilibrium design}. Specifically, we use \emph{reward schemes} to incentivise players so that the Nash equilibria of the game satisfy the desired property.
In our model, agents are represented as concurrently executing processes that operate synchronously. Each agent receives an integer payoff for every state visited by the overall system. The total payoff an agent receives over an infinite computation path is defined as the \textit{mean-payoff} over that path. While agents naturally seek to maximise their individual mean-payoffs, the designer of the reward scheme aims to ensure that a specific temporal logic formula is satisfied on some or all Nash equilibria of the resulting game. 

\begin{wraptable}[17]{r}{0.71\textwidth}
	\vspace{-1em}
	\begin{center}
		\def\arraystretch{1}
		\begin{tabular}{@{}l l l@{}}
			\toprule
			& \LTL Spec. & \GRone Spec. \\
			\hline\\[-1em]
			\textsc{Weak Impl.} & \pspace-c (Thm. \ref{thm:weak-ltl}) & \np-c (Thm. \ref{thm:weak-gr1}) \\
			\textsc{Strong Impl.} & \pspace-c (Cor. \ref{cor:strong-ltl}) & $ \SigmaPTwo $-c (Thm. \ref{thm:strong-gr1}) \\
			\optwi & \fpspace-c (Thm. \ref{thm:optwi-ltl}) & $ \FP^{\np} $-c (Thm. \ref{thm:optwi-gr1}) \\
			\optsi & \fpspace-c (Thm. \ref{thm:optsi-ltl}) & $ \FP^{\SigmaPTwo} $-c (Thm. \ref{thm:optsi-gr1}) \\
			\exactwi & \pspace-c (Cor. \ref{cor:exactwi-ltl}) & $ \DP $-c (Cor. \ref{cor:exactwi-gr1}) \\
			\exactsi & \pspace-c (Cor. \ref{cor:exactsi-ltl}) & $ \DPTwo $-c (Cor. \ref{cor:exactsi-gr1}) \\
			\uoptwi & \pspace-c (Cor. \ref{cor:uoptwi-ltl}) & $ \DeltaPTwo $-c (Cor. \ref{cor:uoptwi-gr1}) \\
			\uoptsi & \pspace-c (Cor. \ref{cor:uoptsi-ltl}) & $ \DeltaPThree $-c (Cor. \ref{cor:uoptsi-gr1}) \\
			
			\utwi & \pspace-c (Thm. \ref{thm:threshold-usw-ltl}) & \np-c (Thm. \ref{thm:threshold-usw-gr1}) \\
			\utsi & \pspace-c (Thm. \ref{thm:threshold-usw-ltl}) & $ \SigmaPTwo $-c (Thm. \ref{thm:threshold-usw-gr1}) \\
			
			\etwi & \pspace-c (Thm. \ref{thm:threshold-rsw-ltl}) & \np-c (Thm. \ref{thm:threshold-rsw-gr1}) \\
			\etsi & \pspace-c (Thm. \ref{thm:threshold-rsw-ltl}) & $ \SigmaPTwo $-c (Thm. \ref{thm:threshold-rsw-gr1}) \\
			\bottomrule
		\end{tabular}
	\end{center}
	\caption{Summary of main complexity results.}
	\label{tab:results}
\end{wraptable}

With this model, we assume that the designer -- an external principal -- has a 
finite budget that is available for designing reward schemes, and this budget can be
allocated across agent/state pairs. By allocating this budget
appropriately, the principal can incentivise players away from some
states and towards others. Since the principal has some temporal
logic goal formula, it desires to allocate rewards so that players
are rationally incentivised to choose strategies so that the
principal's temporal logic goal formula is satisfied in the path that would result from
executing the strategies.  For this general problem,
following~\cite{WEKL13}, we identify two variants of the principal's
mechanism design problem, which we refer to as \textsc{Weak
	Implementation} and \textsc{Strong Implementation}. In the
\textsc{Weak} variant, we ask whether the principal can allocate the
budget so that the goal is achieved on \emph{some} computation path that would be
generated by Nash equilibrium strategies in the resulting system; in
the \textsc{Strong} variation, we ask whether the principal can
allocate the budget so that the resulting system has at least one Nash
equilibrium, and moreover the temporal logic goal is satisfied on {\em all} paths that
could be generated by Nash equilibrium strategies. For these two
problems, we consider goals specified by \LTL formulae or
\GRone~formulae \cite{BJPPS12}, give algorithms for each case, and
classify the complexity of the problem. 
While \LTL is a natural language for the specification of properties of concurrent 
and multi-agent systems, \GRone is an \LTL fragment that can be used to 
easily express several prefix-independent properties of computation paths of reactive systems, 
such as $\omega$-regular properties often used in automated formal verification. 
We then go on to study
variations of these two problems, for example considering 
{\em optimality} and {\em uniqueness} of solutions. We also examine a setting in which a (benevolent) principal considers the welfare of the players in the design of a reward scheme. To capture this setting, we introduce two concepts: \textit{utilitarian} and \textit{egalitarian} social welfare measures. We show that, while the problems associated with \LTL specifications are in \pspace (or \fpspace), the ones with \GRone specifications lie within the polynomial hierarchy, thus
making them potentially amenable to more efficient practical implementations. 
Table~\ref{tab:results} summarises the main computational complexity results in the paper.

\paragraph{Structure of the paper} The remainder of this article is structured as follows.
\begin{itemize}
	\item Section~\ref{sec:prelims} presents the relevant background on games, logic, and Nash equilibrium.
	\item In Section~\ref{sec:eqdesign} we formalise the concept of reward schemes.
	\item In Section~\ref{sec:weak} and \ref{sec:strong} we describe the main problems of interest and present the proofs to obtain tight computational complexity bounds.
	\item In Section~\ref{sec:opt-unique} we study variations of the main problems, including optimality and uniqueness of solutions, and show their respective computational complexity classes.
	\item In Section~\ref{sec:sw} we consider two of the most important social welfare measures, and examine the related computational problems.
	
	\item In Section~\ref{sec:conc} we conclude, discuss related work, and propose some directions for further research.
\end{itemize}

\section{Preliminaries}\label{sec:prelims}

\noindent \textbf{Complexity Classes.}
Here we briefly describe the different complexity classes used in this paper.
We assume that the reader is familiar with the classes \np, \pspace and the notation for complexities relative to an oracle within the polynomial hierarchy~\cite{1994-papadimitriou}.
In particular, we assume the following

\begin{itemize}
	\item 
		$\Sigma_0^{\mthfun{P}} = \Pi_0^{\mthfun{P}} = \Delta_0^{\mthfun{P}} = \mthfun{P} = \Delta_1^{\mthfun{P}}$;
		
	\item 
		$\Sigma_{i + 1}^{\mthfun{P}} = \np^{\Sigma_{i}^{\mthfun{P}}}$;
		
	\item 
		$\Pi_{i + 1}^{\mthfun{P}} = \conp^{\Sigma_{i}^{\mthfun{P}}}$;
	
	\item 
		$\Delta_{i + 1}^{\mthfun{P}} = \mthfun{P}^{\Sigma_{i}^{\mthfun{P}}}$.
\end{itemize}
For instance, $\SigmaPTwo = \np^{\np}$ is the class of problems that can be solved in polynomial time by a non-deterministic Turing machine that can invoke an oracle to solve another \np problem.

For a given decision problem in the complexity class $ \calC $, $ \mthfun{F}\calC $ denotes the complexity class of the corresponding \textit{function} problem. Consider, for example, the class $ \np $ and a problem $ \calP \in \np $. The problem of finding a solution to an instance of $ \calP $ is in $ \mthfun{F}\np $.

Finally, $\mathsf{D^P_\mathnormal{i}}$ denotes the class of languages that are the intersection of a language in $\Sigma_{i}^{\mthfun{P}}$ and a language in $\Pi_{i}^{\mthfun{P}}$ (note that this is not the same as $\Sigma_{i}^{\mthfun{P}} \cap \Pi_{i}^{\mthfun{P}}$). Intuitively, this corresponds to the class of problems that require two consecutive and independent calls to a $\Sigma_{i}^{\mthfun{P}}$ procedure and a $\Pi_{i}^{\mthfun{P}}$ procedure.
This is typically used to refer to problems in which the solution is, in a sense, \textit{unique} or \emph{exact}.
For instance, consider the problem:

\begin{center}
	$\mthfun{EXACT}\text{-}\mthfun{CLIQUE} = \set{\tuple{G, k}}{\text{the largest clique in graph $G$ is of size } k}$.
\end{center}

This requires solving CLIQUE for $k$ to decide whether a clique of size $k$ \textit{exists}, and solving CO-CLIQUE for $k+1$ to decide whether a clique of size $k + 1$ does \textit{not exist}. For more in-depth presentation of the class $\mathsf{D^P_\mathnormal{i}}$, see \cite{1994-papadimitriou,PAPADIMITRIOU1984244,Aleksandrowicz2017}.

\vspace{4pt}
\noindent \textbf{Linear Temporal Logic.}
\LTL~\cite{pnueli:77a} extends classical propositional logic with two
operators, $\ltlnext$ (``next'') and $\until$ (``until''), that can be used to express properties of infinite paths.  The syntax of \LTL is defined
with respect to a set $\AP$ of atomic propositions by the following grammar:
$$ \phi ::= 
\mathop\top \mid
p \mid
\neg \phi \mid\phi \vee \phi \mid
\ltlnext \phi \mid
\phi \until \phi
$$
where $p \in \AP$.
As is conventional in the \LTL literature, we introduce some further classical and temporal operators via the following equivalences: \[ \phi_1 \wedge \phi_2 \equiv \neg (\neg \phi_1 \vee \neg \phi_2)  \qquad \phi_1 \to \phi_2 \equiv \neg \phi_1 \vee \phi_2 \qquad \sometime \phi \equiv \top \until \phi \qquad \always \phi \equiv \neg \sometime \neg \phi\]

We interpret formulae of \LTL with respect to pairs $(\alpha,t)$, where $\alpha \in (2^{\AP})^\omega$ is an infinite sequence of atomic proposition evaluations, indicating which propositional variables are true in every time point, and $t \in \Nat$ is a
temporal index into $\alpha$.
Formally, the semantics of \LTL formulae is given by the following rules:
$$
\begin{array}{lcl}
(\alpha,t)\models\mathop\top	\\
(\alpha,t)\models p 				&\text{ iff }&	p\in \alpha_t\\
(\alpha,t)\models\neg \phi			&\text{ iff }&   \text{it is not the case that $(\alpha,t) \models \phi$}\\
(\alpha,t)\models\phi \vee \psi		&\text{ iff }&	\text{$(\alpha,t) \models \phi$  or $(\alpha,t) \models \psi$}\\
(\alpha,t)\models\ltlnext\phi			&\text{ iff }&	\text{$(\alpha,t+1) \models \phi$}\\
(\alpha,t)\models\phi\until\psi	&\text{ iff }&   \text{for some $t' \geq t: \ \big((\alpha,t') \models \psi$  and }\\
&&\quad\text{for all $t \leq t'' < t': \ (\alpha,t'') \models \phi \big)$.}\\
\end{array}
$$

If $(\alpha,0)\models\phi$, we write $\alpha\models\phi$ and say that
\emph{$\alpha$ satisfies~$\phi$}.

\vspace*{4pt} 
\noindent \textbf{General Reactivity of rank 1.}
The language of \emph{General Reactivity of rank 1}, denoted $\GRone$, is the fragment of \LTL given by  formulae written in the following form~\cite{BJPPS12}:
$$
(\always \sometime \psi_1 \wedge \ldots \wedge \always \sometime \psi_m) \to (\always \sometime \phi_1 \wedge \ldots \wedge \always \sometime \phi_n)
\text{,}
$$
where each subformula $\psi_i$ and $\phi_i$ is a Boolean combination of atomic propositions.

\vspace{4pt}
\noindent \textbf{Mean-Payoff.}
For a sequence $r \in \mathbb{R}^\omega$, let $\MP(r)$ be
the \emph{mean-payoff} value of $r$, that is, 
\[ \MP(r) = \lim \inf_{n \to \infty} \avg_n(r) \]
where, for $n \in \mathbb{N}\setminus\{0\}$, we define
$\avg_n(r) = \frac{1}{n}\sum_{j=0}^{n-1} r_j$, with $r_j$ the $(j\!+\!1)$th element of $r$. 

\vspace*{4pt} \noindent \textbf{Arenas.}
An \emph{arena} is a tuple
$ A = \tuple{\Ag,  \Ac, \St, s_0, \trnFun, \labFun} $ 
where $\Ag$, $\Ac$, and $\St$ are finite non-empty sets of \emph{players} (write $N = \card{\Ag}$), \emph{actions}, and \emph{states}, respectively; if needed, we write $ \Ac_i(s) $, to denote the set of actions available to player $ i $ at $ s $; $s_0 \in \St$ is the \emph{initial state}; $\trnFun : \St \times \AcProf \rightarrow \St$ is a \emph{transition function} mapping each pair consisting of a state $s \in \St$ and an \emph{action profile} $\jact \in \AcProf = \Ac^{\Ag}$, one for each player, to a successor state; and $\labFun: \St \to 2^{\AP}$ is a labelling function, mapping every state to a subset of \emph{atomic propositions}.

We sometimes call an action profile $\jact = (\act_{1}, \dots, \act_{n}) \in \AcProf$ a \emph{decision}, and denote $\act_i$ the action taken by player $i$.
We also consider \emph{partial} decisions.
For a set of players $C \subseteq \Ag$ and action profile $\jact$, we let $\jact_{C}$ and $\jact_{-C}$ be two tuples of actions, respectively, one for all players in $C$ and one for all players in $\Ag \setminus C$.
We also write $\jact_{i}$ for $\jact_{\{i\}}$ and $ \jact_{-i} $ for $ \jact_{\Ag \setminus \{i\}} $.
For two decisions $\jact$ and $\jact'$, we write $(\jact_{C}, \jact_{-C}')$ to denote the decision where the actions for players in $ C $ are taken from $\jact$ and the actions for players in $ \Ag \setminus C $ are taken from $\jact'$.

A \emph{path} $\pi = (s_0, \jact^0), (s_1, \jact^1), \ldots$ is an infinite sequence in $(\St \times \AcProf)^{\omega}$ such that $\trnFun(s_k, \jact^k) = s_{k + 1}$ for all $k$.
In particular, $\jact^k_i$ is the action of player $i$ in step $k$.
Sometimes, we call the single iteration $(s_k, \jact^k)$ a \emph{configuration}.

Paths are generated in the arena by each player~$i$ selecting a {\em
	strategy} $\strElm_i$ that will define how to make choices over
time.  We model strategies as finite state machines with output.
Formally, for arena $A$, a strategy
$\strElm_{i} = (Q_{i}, q_{i}^{0}, \delta_i, \tau_i) $ for player $i$
is a finite state machine with output (a transducer), where $Q_{i}$ is
a finite and non-empty set of \emph{internal states}, $ q_{i}^{0} $ is
the \emph{initial state},
$\delta_i: Q_{i} \times \AcProf \rightarrow Q_{i} $ is a deterministic
\emph{internal transition function}, and
$\tau_i: Q_{i} \rightarrow \Ac_i$ an \emph{action function}. Let $\StrSet_i$ be the set of strategies for player $i$. Note that this definition implies that strategies have \textit{perfect information} and \textit{finite memory} (although we impose no bounds on memory size).

A \emph{strategy profile} $\strpElm = (\strElm_1, \dots, \strElm_n)$ is a vector of strategies, one for each player.
As with actions, $\strpElm_{i}$ denotes the strategy assigned to player $i$ in profile $\strpElm$.
Moreover, by $(\strpElm_{B}, \strpElm'_{C})$ we denote the combination
of profiles where players in disjoint $B$ and $C$ are assigned their corresponding strategies in $\strpElm$ and $\strpElm'$, respectively.
Once a state $s$ and profile $\strpElm$ are fixed, the game has an \emph{outcome}, a path in $A$, denoted by $\pi(\strpElm, s)$. 
Because strategies are deterministic, $\pi(\strpElm, s)$ is the unique path induced by $\strpElm$, that is, the sequence $s_0, s_1, s_2, \ldots$ such that 
\begin{itemize}
	\item $s_{k + 1} = \trnFun (s_k, (\tau_1(q^k_1), \ldots, \tau_n(q^k_n)))$, and 
	\item $q^{k + 1}_i = \delta_i(s^k_i, (\tau_1(q^k_1), \ldots, \tau_n(q^k_n)))$, for all $k \geq 0$. 
\end{itemize}
Furthermore, we simply write $ \pi(\strpElm) $ for $ \pi(\strpElm,s_0) $.

Arenas define the dynamic structure of games, but lack a feature that is essential for game theory: \emph{preferences}, which give games their strategic structure.
A \emph{multi-player game} is obtained from an arena $A$ by
associating each player with a \emph{goal}, which represents that player's preferences.
We consider multi-player games with $\MP$ goals.
%
A multi-player \MP game is a tuple
$\Game = \tuple{A, (\wFun_{i})_{i \in \Ag}}$, where $A$ is an
arena and $\wFun_{i}: \St \to \SetZ$ is a function mapping, for every player~$i$, every state
of the arena into an integer number.  
%
In any game with arena $A$, a path $\pi$ in $A$ induces a sequence $\lambda(\pi) = \lambda(s_0) \lambda(s_1) \cdots$ of sets of atomic propositions; if, in addition, $A$ is the arena of an \MP game, then, for each player~$i$, the sequence $\wFun_i(\pi) = \wFun_i(s_0) \wFun_i(s_1) \cdots$ of weights is also induced. 
Unless stated otherwise, for a game $ \Game $ and a path $\pi$ in it, the payoff of player~$i$ is $\pay_i(\pi) = \MP(\wFun_{i}(\pi))$.

\vspace*{4pt} \noindent \textbf{Nash equilibrium.}
Using payoff functions, we can define the game-theoretic concept of Nash equilibrium~\cite{OR94}. 
For a multi-player game $\Game$, a strategy profile
$\strpElm$ is a \emph{Nash equilibrium} of~$\Game$ if, for every player~$i$ and strategy $\strElm'_i$ for player~$i$, we have
$$
\pay_i(\pi(\strpElm))	\geq	\pay_i(\pi((\strpElm_{-i},\strElm'_i))) \ . 
$$
Let $\NE(\Game)$ be the set of Nash equilibria of~$\Game$. Observe that in a given game, there may be more than one Nash equilibria, and as such, different equilibrium outcomes may behave differently---some of which may not be desirable. This was intuitively illustrated in Example~\ref{ex:grid1}. To frame this problem more appropriately within our framework, consider the following example modified from Example~\ref{ex:grid1}.


\begin{example}\label{ex:grid-mp}
	Consider the same setting as in Example~\ref{ex:grid1}, but with the tasks of the robots set as follows. The task for robot $ \bigcirc $ (resp. $ \blacksquare $) is to visit $ (0,0) $ and $ (2,1) $ (resp. $ (2,0) $ and $ (0,1) $) in an alternating fashion infinitely often. Each movement costs 1 unit of energy and gives a payoff of -1. When a robot reaches a target corner after visiting the other one, it gets 3. Each robot wants to maximise the mean-payoff of its infinite run.
\end{example}

\extsTwo

We can model the movements and payoff function of robot $ \bigcirc $ as a transition system in Figure \ref{fig:ex1tsTwo}. 
The vertices are marked with $ (x,y)_f $, where $ f \in \{ 0,1 \} $ is a flag to mark the last corner robot $ \bigcirc $ visited (0 for $ (0,0) $ and 1 for $ (2,1) $.) We set the payoff function for robot $ \bigcirc $ as follows. $ \wFun_{\bigcirc}(e_i) = -1 $ for $ i \in \{1,2,7,8,9,10\} $ and $ \wFun_{\bigcirc}(e_j) = 2 $ for $ j \in \{3,4,5,6\} $\footnote{For clarity in presentation, we have excluded locations $ (0,1) $ and $ (2,0) $ because they are not part of any shortest routes. Additionally, we have placed the payoffs on the edges rather than the vertices. However, it is easy to transform the transition system and push the payoffs to the vertices.}.  We can model the movements and payoff function of robot $ \blacksquare $ in a similar manner. 
The graphical representation of the game can be obtained by taking the cross product of the two transition systems.


From the game we have obtained, we can see that there are many Nash equilibrium runs. We define ``bad'' Nash equilibrium run as one in which the robots occupy the same location simultaneously and infinitely often. From the set of all Nash equilibrium runs, some are bad and some are not. For example, consider a run where robot $ \bigcirc $'s sequence of moves is $ ((0,0)_0(1,0)_0(2,1)_1(1,0)_1)^{\omega} $ and robot $ \blacksquare $'s is $ ((2,0)_0(1,1)_0(0,1)_1(1,1)_1)^{\omega} $. This is a Nash equilibrium run since both robots get the mean-payoff of $ \frac{1}{2} $ and cannot obtain better rewards by changing their actions. Furthermore, it is not a bad Nash equilibrium run, since the robots never simultaneously occupy the same position. Now, consider a different scenario where robot $ \blacksquare $ plays the same strategy as in the previous run, but robot $ \bigcirc $'s sequence of moves is $ ((0,0)_0(1,1)_0(2,1)_1(1,1)_1)^{\omega} $. The mean-payoffs for both robots are still $ \frac{1}{2} $, making this a Nash equilibrium run. However, in this run, the robots will occupy $ (1,1) $ simultaneously and infinitely often, making it a bad Nash equilibrium run.

From a system design perspective, we want to eliminate such bad Nash equilibrium runs. One way to achieve this is by modifying the payoff function for each robot so that the resulting set of equilibria does not include bad runs. We can do this by providing rewards to the robots in order to ``nudge'' them into taking certain paths. 
Consider again the payoff function of robot $ \bigcirc $ shown in Figure~\ref{fig:ex1tsTwo}. Suppose we provide rewards for $ e_3 $ and $ e_5 $, 1 unit of payoff each. Thus we have
$ \wFun_{\bigcirc}(e_3) = 3 $,
$ \wFun_{\bigcirc}(e_5) = 3 $.
%
Now consider the run of robot $ \bigcirc $ as follows: $ h ((0,0)_0(1,0)_0(2,1)_1(1,0)_1)^{\omega} $, where $ h $ is a finite prefix. This run is a Nash equilibrium that results in a mean-payoff of 1 for robot $ \bigcirc $. In fact, in \textit{every} Nash equilibrium, robot $ \bigcirc $'s run corresponds to this type of run, since otherwise, the robot will get a mean-payoff of $ < 1 $. Similarly, we can design a reward scheme for robot $ \blacksquare $ that will result in $ \blacksquare $ always choosing the run $ h ((2,0)_0(1,1)_0(0,1)_1(1,1)_1)^{\omega} $ in every Nash equilibrium. By combining these payoff functions, we obtain a new payoff function that prevents the system from getting stuck in bad equilibria.

In the next section, we formalise the problem of designing payoff functions using {rewards} to achieve desirable Nash equilibrium runs.

\section{From Mechanism Design to Equilibrium Design}
\label{sec:eqdesign}
We now describe a method for modifying the payoff functions of players in a given game to achieve desirable Nash equilibrium runs. As discussed in the introduction, this problem is closely related to the well-known concept of {\em mechanism design} in game theory. 
Consider a system with multiple agents, represented by the set \Ag. Each agent $ i \in \Ag $ aims to maximise its payoff $ \pay_i(\cdot) $.
As in a mechanism design problem, we assume there is an external \textit{principal} who has a goal $ \phi $ that it
wants the system to satisfy. To accomplish this, the principal seeks to incentivise the agents to act collectively and rationally to bring about $ \phi $. In our
model, incentives are given by \textit{reward schemes} and goals by temporal logic formulae. 

\vspace{4pt}
\noindent \textbf{Reward Schemes:} A reward scheme defines \textit{additional}
imposed payoff over those given by the weight function $ \wFun $. 
While the weight function $ \wFun$ is fixed for any
given game, the principal is assumed to be at liberty to define a reward 
scheme as they see fit. Since agents will seek to maximise their overall rewards,
the principal can incentivise agents to visit certain states and avoid others.  If the reward scheme is designed correctly, the agents are incentivised to choose a strategy profile $ \vec{\sigma} $
such that $ \pi(\strpElm) \models \phi $.
Formally, we model a reward scheme as a function $ \kappa: (\Ag \to \St) \to \mathbb{N} $, where the intended 
interpretation is that $ \kappa(i)(s) $ is the reward in the form of a natural number $ k \in \mathbb{N}$ that would be imposed on player $ i $ if such a player visits state $ s \in \St $. For instance, if we have $ \wFun_{i}(s) = 1 $ and $ \kappa(i)(s) = 2 $, then player~$i$ gets $1 + 2 = 3 $ for visiting such a state. For simplicity, hereafter we write $ \kappa_i(s) $ instead of $ \kappa(i)(s)$ for the reward for player~$i$. 

Notice that having an unlimited fund for a reward scheme would make some problems trivial, as the principal can always incentivise players to satisfy $ \phi $ (provided that there is a path in $ A $ satisfying $ \phi $). A natural and more interesting setting is that the principal is given a constraint in the form of \textit{budget} $\beta\in\mathbb{N}$. The principal then can only spend within the budget limit. To make this clearer, we first define the \textit{cost} of a reward scheme $ \kappa $ as follows.

\begin{definition}\label{def:cost}
	Given a game $ \Game $ and reward scheme $ \kappa $, we define \[ \cost(\kappa) = \sum_{i \in \Ag} \sum_{s \in \St} \kappa_i(s) \text{.}\]
\end{definition}

We say that a reward scheme $ \kappa $ is \textit{admissible} if it does not exceed the budget~$\beta$,  that is, if $ \cost(\kappa) \leq \beta $. Let $ \K(\Game,\beta) $ denote the set of admissible reward schemes over $ \Game $ given budget $ \beta \in \mathbb{N} $. Thus we know that for each $ \kappa \in \K(\Game,\beta)$ we have $\cost(\kappa) \leq \beta $. We write $ (\Game,\kappa) $ to denote the resulting game after the application of reward scheme $ \kappa $ on game $ \Game $. Formally, we define the application of some reward scheme on a game as follows.

\begin{definition}\label{def:apply}
	Given a game $ \Game = \tuple{A,(\wFun_{i})_{i \in \Ag}} $ and an admissible reward scheme $ \kappa $, we define $ (\Game,\kappa) = \tuple{A, (\wFun'_{i})_{i \in \Ag}} $, where $ \wFun'_{i}(s) = \wFun_{i}(s) + \kappa_i(s)$, for each $ i \in \Ag $ and $ s \in \St $.
\end{definition}

We now come to the main question(s) that we consider in the remainder of the paper. We ask whether the principal can find a reward scheme that will incentivise players to collectively choose a rational outcome (a Nash equilibrium) that satisfies its temporal logic goal $ \phi $. We call this problem {\em equilibrium design}. Following~\cite{WEKL13}, we define two variants of this problem, a {\em weak} and a {\em strong} implementation of the equilibrium design problem. The formal definition of the problems and the analysis of their respective computational complexity are presented in the next sections.

\begin{remark}
	For the rest of the paper, we assume that weights and reward schemes are using a binary representation for their values.
	This is the standard way of considering them in the context of mean-payoff games~\cite{ZP96,UW11}.
\end{remark}

\section{Equilibrium Design: Weak Implementation}\label{sec:weak}
In this section, we study the weak implementation of the equilibrium design problem, a logic-based computational variant of the principal's mechanism design problem in game theory. We assume that the principal has full knowledge of the game $ \Game $ under consideration, that is, the principal uses all the information available of $ \Game $ to find the appropriate reward scheme, if such a scheme exists. We now formally define the weak variant of the implementation problem, and study its respective computational complexity, first with respect to goals (specifications) given by \LTL formulae and then with respect to \GRone formulae. 

Let $ \WI(\Game,\phi,\beta) $ denote the set of reward schemes over $ \Game $ given budget $ \beta $ that satisfy a formula $ \phi $ in at least one path $ \pi $ generated by $ \strpElm \in \NE(\Game) $. Formally
$$ \WI(\Game,\phi,\beta) = \{ \kappa \in \K(\Game,\beta) : \exists \strpElm \in \NE(\Game,\kappa)~\textrm{s.t.}~\pi(\strpElm) \models \phi \}. $$

\begin{definition}[\weak]
	Given a game $\Game$, formula $\varphi$, and budget $ \beta $:
	\begin{center}
		Is it the case that $ \WI(\Game,\phi,\beta) \neq \varnothing $?
	\end{center}
\end{definition}

In order to solve \weak, we first characterise the Nash equilibria of a multi-player concurrent game in terms of punishment strategies. To do this in our setting, we recall the notion of secure values for mean-payoff games~\cite{UW11}.

For a player $i$ and a state $s \in \St$, by $\pun_i(s)$ we denote the
punishment value of $i$ over $s$, that is, the maximum payoff that $i$
can achieve from $s$, when all other players behave adversarially.
Note that the value $\pun_{i}(s)$ corresponds to the one of a two-player zero-sum mean-payoff game~\cite{ZP96}, where the coalition $-i = \Ag \setminus \{i\}$ is playing adversarially against $i$.
Thus, computing $\pun_{i}(s)$ amounts to computing the winning value of $i$ in such two-player zero-sum mean-payoff game, which can then be done in $\np \cap \conp$.
Also, note that the coalition $\Ag \setminus \{i\}$ can achieve the optimal value of the game using \emph{memoryless} strategies.

%
Then, for a player $i$ and a value $z \in \SetR$, a pair $(s, \jact)$ is $z$-secure for player~$i$ if $\pun_i(\trnFun(s, (\jact_{-i}, \act'_i))) \leq z$ for every $\act'_i \in \Ac$. Write $\pun_{i}(\Game)$ for the set of punishment values for player~$i$ in $\Game$.

\begin{theorem}
	\label{thm:pthfinding}
	For every \MP game $\Game$ and ultimately periodic path $\pi = (s_0, \jact_{0}), (s_1, \jact^{1}), \ldots $, the following are equivalent:
	
	\begin{enumerate}
		\item 
		There is $\strpElm \in \NE(\Game)$ such that $\pi = \pi(\strpElm, s_0)$;
		
		\item 
		There exists $ {z} \in \SetR^{\Ag}$, where $z_{i} \in \pun_{i}(\Game)$ such that, for every $i \in \Ag$
		
		\begin{enumerate}
			\item 
			for all $k \in \SetN$, the pair $(s_k, \jact^k)$ is $z_i$-secure for $i$, and 
			
			\item 
			$z_i \leq \pay_i(\pi)$.
		\end{enumerate}
	\end{enumerate}
	
\end{theorem}

\begin{proof}
	For (1) implies (2): Let $ z_i $ be the largest value player $ i $ can get by deviating from $ \pi $. Let $ k \in \SetN $ be such that $ z_i = \pun_i(\trnFun(s_k,(\jact_{-i},\act'_i))) $. Suppose further that $ \pay_i(\pi) < z_i $. Thus, player $ i $ would deviate at $ s_k $, which is a contradiction to $ \pi $ being a path induced by a Nash equilibrium.
	
	For (2) imples (1): Define strategy profile $ \strpElm $ that follows $ \pi $ as long as no-one has deviated from $ \pi $.
	In such a case where player $ i $ deviates on the $k$-th iteration, the strategy profile $\strpElm_{-i}$ starts playing the $z_i$-secure strategy for player $i$ that guarantees the payoff of player $i$ to be less than $z_i$. Therefore, we have $\pay_i(\pi(\strpElm_{-i},\strElm'_{i})) \leq z_i \leq \pay_i(\pi)$, for every possible strategy $\strElm'_{i}$ of player $i$ (the second inequality is due to condition $2(b)$). Thus, there is no beneficial deviation for player $ i $ and $ \pi $ is a path induced by a Nash equilibrium.
	Indeed, by contradiction, assume for player $i$ that a strategy $\strElm_{i}'$ is a beneficial deviation from the strategy profile $\strpElm$.
	Then, we would have $z_ i \leq \pay_i(\strpElm) < \pay_i(\strpElm_{-i}, \strElm_{i}') < z_i$, the last inequality following from the fact that $\strpElm_{-i}$ is $z_i$-secure for coalition $-i$.
	Clearly, the sequence of inequality makes it a contradiction.
\end{proof}

The characterisation of Nash Equilibria provided in Theorem~\ref{thm:pthfinding} will allow us to turn the \weak problem into a {\em path finding} problem over $(\Game,\kappa)$. 
On the other hand, with respect to the budget $\beta$ that the principal has at its disposal, the definition of reward scheme function $\kappa$ implies that the size of $ \K(\Game,\beta) $ is bounded, and particularly, it is bounded by $\beta$ and the number of agents and states in the game $\Game$, in the following way. 

\begin{proposition}\label{prop:kbound}
	Given a game $ \Game $ with $\card{N}$ players and $ \card{\St} $ states and budget $\beta$, it holds that 
	$$ \card{\K(\Game,\beta)} = \frac{\beta + 1}{m} \binom{\beta + m}{\beta + 1}\text{,}$$
	with $m = \card{N \times \St}$ being the number of pairs of possible agents and states.
\end{proposition}

\begin{proof}
	
	For a fixed budget $b$, the number of reward schemes of budget exactly $b$ corresponds to the number of \emph{weak compositions} of $b$ in $m$ parts, which is given by $\binom{b+m-1}{b}$~\cite{HT09}.
	Therefore, the number of reward schemes of budget at most $\beta$ is the sum
	$$
	\card{\K(\Game,\beta)} = \sum_{b=0}^{\beta}\binom{b+m-1}{b}\text{.}
	$$
	
	We now prove that 
	$$\sum_{b=0}^{\beta}\binom{b+m-1}{b} = \frac{\beta + 1}{m} \binom{\beta + m}{\beta + 1}\text{.}$$
	
	By induction on $\beta$, as base case, for $\beta = 0$, we have that 
	$$
	\binom{\beta + m - 1}{\beta} = 1 = \frac{\beta + 1}{m} \binom{\beta +m}{\beta +1}\text{.}
	$$
	
	For the inductive case, let us assume that the assertion hold for some $\beta$ and let us prove for $\beta + 1$.
	We have the following:
	$$\sum_{b=0}^{\beta + 1}\binom{b+m-1}{b} = \sum_{b=0}^{\beta}\binom{b+m-1}{b} + \binom{\beta + m - \cancel{1} + \cancel{1}}{\beta + 1} = \frac{\beta + 1}{m} \binom{\beta + m}{\beta + 1} + \binom{\beta + m}{\beta + 1} \text{.}$$
	
	Therefore we have
	\begin{align*}
	\frac{\beta + 1}{m} \binom{\beta + m}{\beta + 1} + \binom{\beta + m}{\beta + 1} & = \binom{\beta + m}{\beta + 1} \left(\frac{\beta + 1}{m} + 1\right) = \\
	\binom{\beta + m}{\beta + 1} \frac{\beta + 1 + m}{m} 
	& = \frac{\beta + 1 + m}{m} \cdot \frac{(\beta + m)!}{(\beta + 1)!(\cancel{\beta} + m - \cancel{\beta} - 1)!} = \\
	\frac{(\beta + m + 1)!}{(\beta + 1)!m!} = \frac{(\beta + m + 1)!}{(\beta + 1)!m!}\cdot \frac{\beta + 2}{\beta + 2} \cdot \frac{m}{m} & = \frac{\beta + 2}{m} \cdot \frac{(\beta + m + 1)!}{(\beta + 2)! (m - 1)!} = \\
	\frac{\beta + 2}{m} \cdot \frac{(\beta + m + 1)!}{(\beta + 2)! (\beta + m + 1 - \beta - 2)!} & = \frac{\beta + 2}{m} \binom{\beta + m + 1}{\beta + 2}
	\end{align*}
	This proves the assertion.
\end{proof}

From Proposition~\ref{prop:kbound} we derive that the number of possible reward schemes is {\em polynomial} in the budget $\beta$ and {\em singly exponential} in both the number of agents and states in the game.

At this point, solving \weak can be done with the Algorithm~\ref{alg:weak}.
Note that, in Line 4, the algorithm builds a game $(\Game,\kappa){[z]}$, which is obtained from $(\Game,\kappa)$ by removing the states $s$ such that $\pun_{i}(s) \leq z_i$ for some player $i$, and transitions $(s, \jact_{-i})$ that are not $z_i$ secure for player $i$.

\begin{algorithm}[t]
	\caption{Algorithm for \weak}
	\begin{algorithmic}[1]
		\Statex \textbf{input:} $ (\Game, \phi, \beta)  $
		\State \begin{varwidth}[t]{\linewidth}
				{guess:}\\
				\phantom{text} \textbullet~ a reward scheme $ \kappa \in \K(\Game,\beta) $\\
				\phantom{text} \textbullet~ 
				a state $ s \in \St $ for every player $ i \in \Ag $\\
				\phantom{text} \textbullet~ punishment memoryless strategies $ (\strpElm_{-1},\dots,\strpElm_{-n}) $ for all players $ i \in \Ag $
			\end{varwidth} \label{proc:guesses}
		
		\State Compute $ (\Game,\kappa) $ \label{proc:compute-gk}
		
		\State Compute the vector $ z $ where $ z_i = \pun_i(s) $ w.r.t. the punishment strategy $ \strpElm_{-i} $ \label{proc:compute-z}
		
		\State Build $(\Game,\kappa){[z]}$ \label{proc:compute-gz}
		
		\If{There exists $ \pi $ in $ (\Game,\kappa){[z]}$ such that $ \pi \models \phi $ and $ z_i \leq \pay_i(\pi) $ for every player $ i \in \Ag $} \label{proc:find-path}
		
			\State \textbf{return} YES
		
		\Else
			\State \textbf{return} NO
		\EndIf
		
	\end{algorithmic}
	\label{alg:weak}
\end{algorithm}

%
%
%
%
%

\begin{theorem}\label{thm:weak-ltl}
	\weak with \LTL specifications is \pspace-complete.
\end{theorem}

\begin{proof}
	
	Firstly, notice that the correctness of Algorithm~\ref{alg:weak} directly follows from the characterisation provided in Theorem~\ref{thm:pthfinding} and the definition of \weak.
	Specifically, the path \( \pi \) in line~\ref{proc:find-path} corresponds to a NE run of $ (\Game,\kappa) $. In other words, \( \pi = \strpElm \in \NE((\Game,\kappa)) \) because it satisfies condition (2) in Theorem~\ref{thm:pthfinding}. Moreover, \( \pi \) also satisfies the property \( \phi \). According to the definition of \weak, this provides a witness for the non-emptiness of \( \WI(\Game,\phi,\beta) \). Therefore, we can conclude that Algorithm~\ref{alg:weak} correctly solves \weak.
	 
	Regarding the complexity, observe the following.
	Since the set $ \K(\Game,\beta) $ is finitely bounded (Proposition \ref{prop:kbound}), and punishment strategies only need to be memoryless, thus also finitely bounded, clearly line~\ref{proc:guesses} can be guessed non-deterministically. Moreover, each of the guessed elements is of polynomial size, thus this step can be done (deterministically) in polynomial space. 
	Line~\ref{proc:compute-gk} clearly can be done in polynomial time.
	Line~\ref{proc:compute-z} can also be done in polynomial time since, given $ (\strpElm_{-1},\dots,\strpElm_{-n}) $, we can compute the vector $ z = (z_1,\dots,z_n) $ by solving $|\Ag|$ number of \textit{one-player mean-payoff games}, one for each player $ i $---this can be done in polynomial time for each $ i $~\cite[Thm.~6]{ZP96}---then set $ z_i = \pun_{i}(s) $ (see Theorem~\ref{thm:pthfinding}). For line~\ref{proc:find-path}, we will use Theorem~\ref{thm:pthfinding} and consider two cases, one for \LTL specifications and one for \GRone specifications. We first consider the case with \LTL specifications, and in the next subsection with \GRone specifications. For \LTL specifications, consider the formula
	$$
	\phi_{\WI} := \phi \wedge \bigwedge_{i \in \Ag} (\MP(i) \geq z_i)
	$$
	written in $\LTLlim$~\cite{BCHK14}, an extension of \LTL where statements about mean-payoff values over a given weighted arena can be made.%
	\footnote{The formal semantics of $\LTLlim$ can be found in ~\cite{BCHK14}. We prefer to give only an informal description here.}
	%
	The semantics of the temporal operators of $\LTLlim$ is just like the one for \LTL over infinite computation paths $\pi = s_0,s_1,s_3.\ldots$. On the other hand, the meaning of $\MP(i) \geq z_i$ is simply that such an atomic formula is true if, and only if, the mean-payoff value of $\pi$ with respect to player~$i$ is greater or equal to $z_i$, a constant real value; that is, $\MP(i) \geq z_i$ is true in $\pi$ if and only if $\pay_i(\pi) = \MP(\wFun_{i}(\pi))$ is greater or equal than constant value $z_i$. 
	Formula $ \phi_{\WI} $ corresponds exactly to $ 2(b) $ in Theorem \ref{thm:pthfinding}. Furthermore, since every path in $ (\Game,\kappa){[z]} $ satisfies condition $ 2(a) $ of Theorem \ref{thm:pthfinding}, every computation path of $ (\Game,\kappa){[z]} $ that satisfies $ \phi_{\WI} $ is a witness to the \weak problem.
	
	Therefore, membership follows from the algorithm and the fact that model checking for $\LTLlim$ is \pspace-complete~\cite{BCHK14}.
	
	Hardness follows from the fact that \LTL model checking is a special case of \weak. To see this, consider an instance of \weak $ (\Game, \phi, \beta) $ where $ \Ag = \{1\}, \beta = 0 $ and for every $ s \in \St $, $ \wFun_1(s) = 0 $. Clearly, $ \WI(\Game,\phi,\beta) \neq \varnothing $ if and only if the LTL formula $ \phi $ is satisfied in the underlying arena $ A $ of $ \Game $, which constitutes model checking $ \phi $ against $ A $.
\end{proof}

\begin{remark}
	Note that the formula $\phi_{\WI}$ in the proof of the theorem is used to check whether there is a path satisfying the formula. We refer to this as ``existential'' $ \LTLlim $ model checking. This notion is not directly addressed in \cite{BCHK14}, where the discussion is centered around ``universal'' model checking.
	However, one can easily be derived from the other by negating the formula and flipping the answer of model checking, all whilst remaining within \pspace.
	We further note that the language $ \LTLlim $ is closed under negation, and (as such) strict inequalities are also expressible.
	Indeed, strict inequalities are also explicitly used in~\cite{BCHK14}.
\end{remark}

\noindent\textbf{Case with $ \GRone $ specifications.} 
One of the main bottlenecks of our algorithm to solve \weak lies in line~\ref{proc:find-path}, where we solve an $ \LTLlim $ model checking problem. To reduce the complexity of our decision procedure, we consider \weak with the specification $ \phi $ expressed in the $ \GRone $ sublanguage of \LTL. With this specification language, we can avoid model checking $ \LTLlim $ in line~\ref{proc:find-path}. Indeed, with \GRone specifications, we can solve line~\ref{proc:find-path} in polynomial time. This is made possible by a linear program (\LP) that we define, drawing inspiration from Kosaraju and Sullivan's technique for detecting zero cycles~\cite{kosaraju1988detecting}. The \LP~yields a solution if and only if line~\ref{proc:find-path} returns true.

\begin{theorem}\label{thm:weak-gr1}
	\weak with $ \GRone $ specifications is \np-complete.
\end{theorem}

\begin{proof}
	For the upper bound, observe that in Algorithm~\ref{alg:weak}, line~\ref{proc:guesses} can be done \textit{non-}\textit{de\-ter\-min\-is\-ti\-cal\-ly} in polynomial time. Furthermore, lines~\ref{proc:compute-gk}--\ref{proc:compute-gz} can be done \textit{deterministically} in polynomial time. Then, in order to solve line~\ref{proc:find-path} we define a linear program of size polynomial in $(\Game,\kappa)$ having a solution if and only if there exists an ultimately periodic path $ \pi $ such that $z_i \leq \pay_i(\pi) $ and satisfies the \GRone specification.
	
	To do this, first recall that $\varphi$ has the following form
	$$
	\varphi \coloneq \bigwedge_{l = 1}^{m} \always \sometime \psi_{l} \to \bigwedge_{r = 1}^{n} \always \sometime \theta_{r}\text{,}
	$$
	and let $V(\psi_{l})$ and $V(\theta_r)$ be the subset of states in $(\Game,\kappa)$ that satisfy the boolean combinations $\psi_{l}$ and $\theta_{r}$, respectively. Observe that property $\varphi$ is satisfied over a path $\pi$ if, and only if, either $\pi$ visits every $V(\theta_r)$ infinitely many times or visits some of the $V(\psi_{l})$ only a finite number of times.
	
	For the game $(\Game,\kappa){[z]}$, let $ W = (V, E, (\wFun_i')_{i \in \Ag})$ be the underlying multi-weighted graph, where $ \wFun_i'(v) = \wFun_i(s) - z_i $ for every $ i \in \Ag $, $ v \in V $, and $ s \in \St $ such that $ v $ corresponds to $ s $. Furthermore, we introduce a variable $x_e$ for every edge $e\in E$, where the value $x_e$ corresponds to the number of times that the edge $e$ is used on a cycle. Let $\src(e) $ and $ \trg(e) $ be the \textit{source} and \textit{target} of the edge $ e $, respectively; $\OUT(v) = \{e \in E : \src(e) = v\}$; and $\IN(v) = \{e \in E : \trg(e) = v\}$.
	
	Consider $\psi_{l}$ for some $1 \leq l \leq m$, and define the linear program $\LP(\psi_{l})$ with the following inequalities and equations:
	\begin{enumerate}[align=left]
		\item[Eq1:]
		$x_e \geq 0$ for each edge $e$ 
		--- an edge cannot be used a negative number of times;
		
		\item[Eq2:]
		$\Sigma_{e \in E} x_e \geq 1$ 
		--- ensures that at least one edge is chosen;
		
		\item[Eq3:]
		for each $i \in \Ag$, $\Sigma_{e \in E} \wFun_i'(\trg(e)) x_e \geq 0$ --- this enforces that the total sum of any solution is non-negative~\footnote{Notice that by using $\wFun_i'(\src(e)) = \wFun_i(\src(e)) - z_i$ we ensure that a non-negative cycle corresponds to paths where agent $i$ ensures a payoff greater or equal than $z_i$, which is the $z$-secure value for them, thus preventing deviations.};
		
		\item[Eq4:]
		$\Sigma_{\trg(e) \cap V(\psi_{l}) \neq \emptyset} x_e = 0$ --- this ensures that no state in $V(\psi_{l})$ is in the cycle associated with the solution;
		
		\item[Eq5:]
		for each $v \in V$, $\Sigma_{e \in \OUT(v)} x_e = \Sigma_{e \in \IN(v)} x_e$  --- this condition says that the number of times one enters a vertex is equal to the number of times one leaves that vertex.	
	\end{enumerate}
	
	Now, by construction, it follows that $\LP(\psi_{l})$ admits a solution if and only if there exists a path $\pi$ in $\Game$ such that $z_i \leq \pay_i(\pi) $ for every player $i$ and visits $V(\psi_{l})$ only {\em finitely many times}.
	Now, consider the linear program $\LP(\theta_{1}, \ldots, \theta_{n})$ defined with the following inequalities and equations: 
	
	\begin{enumerate}[align=left]
		\item[Eq1:]
		$x_e \geq 0$ for each edge $e$ --- 
		an edge cannot be used a negative number of times;
		
		\item[Eq2:]
		$\Sigma_{e \in E} x_e \geq 1$ --- 
		ensures that at least one edge is chosen;
		
		\item[Eq3:]
		for each $i \in \Ag$, $\Sigma_{e \in E} \wFun_i'(\trg(e)) x_e \geq 0$ --- this enforces that the total sum of any solution is non-negative;
		
		\item[Eq4:]
		for all $1 \leq r \leq n$, $\Sigma_{\trg(e) \cap V(\theta_{r}) \neq \emptyset} x_e \geq 1$ --- this ensures that for every $V(\theta_{r})$ at least one state is in the cycle;
		
		\item[Eq5:]
		for each $v \in V$, $\Sigma_{e \in \OUT(v)} x_e = \Sigma_{e \in \IN(v)} x_e$  --- this 
		condition says that the number of times one enters a vertex is equal to the number of times one leaves that vertex.	
	\end{enumerate}

	In this case, $\LP(\theta_{1}, \ldots, \theta_{n})$  admits a solution if and only if there exists a path $\pi$ such that $z_i \leq \pay_i(\pi) $ for every player $i$ and visits every $V(\theta_{r})$ {\em infinitely many times}. As highlighted in \cite{kosaraju1988detecting}, any positive multiple of a solution to the above LP problem also constitutes a valid solution. Consequently, by appropriately scaling any given solution, we can always obtain an integral solution.
	
	Since the constructions above are polynomial in the size of both $(\Game,\kappa)$ and $\phi$, we can conclude it is possible to check in \np the statement that there is a path $\pi$ satisfying $\varphi$ such that $z_i \leq \pay_i(\pi) $ for every player~$i$ in the game if and only if one of the two linear programs defined above has a solution.
	
	For the lower bound, observe that if $\phi \coloneqq \top$  and $ \beta = 0 $, then the problem is equivalent to checking whether the \MP game has a Nash equilibrium, which is \np-hard~\cite{UW11}.
\end{proof}

We now turn our attention to the strong implementation of the equilibrium design problem. As in this section, we first consider \LTL specifications and then \GRone specifications. 

\section{Equilibrium Design: Strong Implementation}\label{sec:strong}
While it may be good news for the principal to find that $ \WI(\Game,\phi,\beta) \neq \varnothing $, it may not be sufficient. Even if there is a desirable Nash equilibrium, it is possible that other equilibria may be undesirable. In such cases, the principal may want to avoid the risk of the system getting stuck in bad equilibria. This motivates us to consider the \textit{strong implementation} variant of equilibrium design. In a strong implementation, we require that \textit{every} Nash equilibrium outcome satisfies the specification~$ \phi $, for a {\em non-empty} set of outcomes. Let $ \SI(\Game,\phi,\beta) $ denote the set of reward schemes $ \kappa $ given budget $ \beta $ over $ \Game $ such that:
\begin{enumerate}
	\item $ (\Game,\kappa) $ has at least one Nash equilibrium outcome,
	\item every Nash equilibrium outcome of $ (\Game,\kappa) $ satisfies $ \phi $.
\end{enumerate}

Formally we define it as follows:
$$
\SI(\Game,\phi,\beta) = \{ \kappa \in \mathcal{K}(\Game,\beta) : \NE(\Game,\kappa) \neq \varnothing \wedge \forall \strpElm \in \NE(\Game,\kappa)~\textrm{s.t.}~\pi(\strpElm) \models \phi \}.
$$
This gives us the following decision problem:
\begin{definition}[\strong]
	Given a game $\Game$, formula $\varphi$, and budget $ \beta $:
	\begin{center}
		Is it the case that $ \SI(\Game,\phi,\beta) \neq \varnothing $?
	\end{center}
\end{definition}

\begin{algorithm}[t]
	\caption{Algorithm for \strong}
	\begin{algorithmic}[1]
		\Statex \textbf{input:} $ (\Game, \phi, \beta)  $
		\State \begin{varwidth}[t]{\linewidth}
			{guess:}\\
			\phantom{text} \textbullet~ a reward scheme $ \kappa \in \K(\Game,\beta) $\\
			\phantom{text} \textbullet~ 
			a state $ s \in \St $ for every player $ i \in \Ag $\\
			\phantom{text} \textbullet~ punishment memoryless strategies $ (\strpElm_{-1},\dots,\strpElm_{-n}) $ for all players $ i \in \Ag $
		\end{varwidth} \label{proc:guesses-strong}
		
		\State Compute $ (\Game,\kappa) $ \label{proc:compute-gk-strong}
		
		\State Compute the vector $ z $ where $ z_i = \pun_i(s) $ w.r.t. the punishment strategy $ \strpElm_{-i} $ \label{proc:compute-z-strong}
		
		\State Build $(\Game,\kappa){[z]}$ \label{proc:compute-gz-strong}
		
		\If{\begin{enumerate}[wide=12pt,nosep,before=\vspace*{-\baselineskip}]
			\item[(a)] there exists $ \pi $ in $ (\Game,\kappa){[z]}$ such that $ z_i \leq \pay_i(\pi) $ for every player $ i \in \Ag $
			
			\textbf{and}
			\item[(b)] there is no $ \pi $ in $ (\Game,\kappa){[z]}$ such that $ z_i \leq \pay_i(\pi) $ for every player $ i \in \Ag $ and $ \pi \models \neg \phi $
		\end{enumerate}
			} \label{proc:find-path-strong}
		
		\State \textbf{return} YES
		
		\Else
		\State \textbf{return} NO
		\EndIf
		
	\end{algorithmic}
	\label{alg:strong}
\end{algorithm}

\strong can be solved with a similar procedure as in Algorithm~\ref{alg:weak} for \weak where lines~\ref{proc:guesses}--\ref{proc:compute-gz} are exactly the same, but with line~\ref{proc:find-path} modified as follows:%
%
%
%
%
%
%
%
%

	Check whether:
	\begin{enumerate}
		\item[(a)] there exists an ultimately periodic path $ \pi $ in $ (\Game,\kappa){[z]}$ such that $ z_i \leq \pay_i(\pi) $ for each $ i \in \Ag $;
		\item[(b)] there is no ultimately periodic path $ \pi $ in $ (\Game,\kappa){[z]}$ such that $ \pi \models \lnot\phi $ and $ z_i \leq \pay_i(\pi) $, for each $ i \in \Ag $.
	\end{enumerate}

Observe that a positive answer to both (a) and (b) above implies that $ \NE(\Game,\kappa) \neq \varnothing $ and for every $ \strpElm \in \NE(\Game,\kappa) $ we have $ \pi(\strpElm) \models \phi $. Thus, $ \kappa \in \SI(\Game,\phi,\beta) $ and as such $ \SI(\Game,\phi,\beta) \neq \varnothing $. The complete algorithm is shown in Algorithm~\ref{alg:strong}. For \LTL specifications, to solve line~\ref{proc:find-path-strong} in Algorithm~\ref{alg:strong}, consider the following \LTLlim formulae that correspond, respectively, to conditions (a) and (b):
\begin{align*}
	\phi_{\exists} &\coloneqq \bigwedge_{i \in \Ag} (\MP(i) \geq z_i);\\
	\phi_{\forall} &\coloneqq \phi_{\exists} \rightarrow \phi.
\end{align*}

Notice that the expression $ \NE(\Game,\kappa) \neq \varnothing $ can be expressed as ``there exists a path $ \pi $ in $ \Game $ that satisfies formula $ \phi_{\exists} $''. On the other hand, the expression $ \forall \strpElm \in \NE(\Game,\kappa)~\textrm{such that}~\pi(\strpElm) \models \phi $ can be expressed as ``for every path $ \pi $ in $ \Game $, if $\pi$ satisfies formula $ \phi_{\exists} $, then $\pi$ also satisfies formula $ \phi $''. Thus, using these two formulae to solve line~\ref{proc:find-path-strong}, we obtain the following result. 

\begin{corollary}\label{cor:strong-ltl}
	\strong with \LTL specifications is \pspace-complete.
\end{corollary}

\begin{proof}
	Using an argument similar to that in Theorem~\ref{thm:weak-ltl}, we can directly infer the correctness of Algorithm~\ref{alg:strong} from Theorem~\ref{thm:pthfinding} and the definition of \strong. In particular, in line~\ref{proc:find-path-strong}, the existence of path $ \pi $ in (a), and the absence of path $ \pi $ in (b), correspond to the non-emptiness of \( \SI(\Game,\phi,\beta) \).
	
	Membership follows from the fact that for line~\ref{proc:find-path-strong}, (a) can be solved by existential \LTLlim model checking, and (b) by universal \LTLlim model checking---both clearly in \pspace by Savitch's theorem. Hardness is similar to the construction in Theorem \ref{thm:weak-ltl}.
\end{proof}

\noindent\textbf{Case with $ \GRone $ specifications.}
Notice that line~\ref{proc:find-path-strong} (a) in Algorithm~\ref{alg:strong} is essentially $ \NE(\Game,\kappa) \neq \varnothing $, that is, checking whether the set of Nash equilibrium in a mean-payoff game is not empty---this can be solved in \np \cite{UW11}. For the (b) part, observe that $$ \forall \strpElm \in \NE(\Game,\kappa)~\textrm{such that}~\pi(\strpElm) \models \phi $$ is equivalent to $$ \neg\exists \strpElm \in \NE(\Game,\kappa)~\text{such that}~\pi(\strpElm) \models \lnot \phi. $$

Thus we have
$$
\lnot\varphi = \bigwedge_{l = 1}^{m} \always \sometime \psi_{l} \wedge \lnot \big( \bigwedge_{r = 1}^{n} \always \sometime \theta_{r} \big)\text{.}
$$
To check the formula above, we modify the \LP~in Theorem \ref{thm:weak-gr1}. Specifically, we modify Eq4 in $\LP(\theta_{1}, \ldots, \theta_{n})$ to encode the $ \theta $-part of $ \lnot \phi $. Thus, we have the following equation in $\LP'(\theta_{1}, \ldots, \theta_{n})$:

\begin{enumerate}[align=left]
	
	\item[Eq4:]
	there exists $ r $, $1 \leq r \leq n$, $\Sigma_{\src(e) \cap V(\theta_{r}) \neq \emptyset} x_e = 0$ --- this condition ensures that at least one set $V(\theta_{r})$ does not have any state in the cycle associated with the solution. 
	
\end{enumerate}
All other equations remain the same.

In this case, $ \LP'(\theta_{1}, \ldots, \theta_{n}) $ has a solution if and only if there is a path $ \pi $ such that $ z_i \leq \pay_i(\pi) $ for every player $ i $ and, for at least one $ V(\theta_{r}) $, its states are visited only \textit{finitely many times}. Thus, we have a procedure that checks if there is a path $ \pi $ that satisfies $ \lnot \phi $ such that $ z_i \leq \pay_i(\pi) $ for every player $ i $, if and only if both linear programs have a solution. Using this new construction, we can now prove the following result. 

\begin{theorem}\label{thm:strong-gr1}
	\strong with \GRone specifications  is $ \SigmaPTwo $-complete.
\end{theorem}

\begin{proof}
	For membership, observe that by rearranging the problem statement, we have the following question:\\
	Check whether the following expression is true
	\begin{align*}
		\exists \kappa \in \K(\Game,\beta)&,\tag{1}\\
		&\exists \strpElm \in \strElm_1 \times \cdots \times \strElm_n,~\text{such that}~\strpElm \in \NE(\Game,\kappa),\tag{2}\\
		~\text{and}&\\
		&\forall \strpElm' \in \strElm_1 \times \cdots \times \strElm_n,~\text{if}~\strpElm' \in \NE(\Game,\kappa)~\text{then}~\pi(\strpElm') \models \phi.\tag{3}
	\end{align*}
	Statement $ (2) $ can be checked in \np (Theorem \ref{thm:pthfinding}). Whereas, verifying statement $ (3) $ is in co\np; to see this, notice that we can rephrase $ (3) $ as follows: $ \not \exists z \in \{ \pun_i(s) : s \in \St \}^{\Ag}  $ such that both $ \LP(\psi_{l}) $ and $ \LP'(\theta_1,\dots,\theta_{n}) $ have a solution in $ (\Game,\kappa){[z]} $. Thus $ \SigmaPTwo $ membership follows.
	
	We prove hardness by a reduction from $ \QSAT_2 $ (satisfiability of quantified Boolean formula with 2 alternations) \cite{1994-papadimitriou}. Let $ \psi(\setx,\sety) $ be an $ n + m $ variable Boolean 3DNF formula, where $ \setx = \{x_1,\dots,x_n \} $ and $ \sety = \{ y_1,\dots, y_n \} $, with $ t_1,\dots,t_k $ terms. Write $ \term_j $ for the set of literals in $ j $-th term and $ \term_j^i $ for the $ i $-th literal in $ \term_j $. Moreover write $ x^j_i $ and $ y^j_i $ for variable $ x_i \in \setx $ and $ y_i \in \sety $ that appears in $ j $-th term, respectively. For instance, if the fifth term is of the form of $ (x_2 \wedge \lnot x_3 \wedge y_4) $, then we have $ \term_5 = \{ x^5_2, x^5_3, y^5_4 \} $ and $ \term_5^1 = x^5_2 $. Let $ \T = \{ \term_i \cap \sety : 1 \leq i \leq k \} $, that is, the set of subset of $ \term_i $ that contains only $ y $-literals.
	
	For a formula $ \psi(\setx,\sety) $ we construct an instance of \strong such that $ \SI(\Game,\phi,\beta) \neq \varnothing $ if and only if there is an $ \vecx \in \{0,1\}^n $ such that $ \psi(\setx,\sety) $ is true for every $ \vecy \in \{0,1\}^m $. Let $ \Game $ be such a game where
	
	\begin{itemize}
		\item $ \Ag = \{1,2\} $,
		\item $ \St = \{ \bigcup_{j \in [1,k]} (\term_j \times \{0,1\}^3) \} \cup \{\T \times \{0\}^3 \} \cup \{ \tuple{\source,\{0\}^3}, \tuple{\sink,\{0\}^3} \} $,
		\item $ s_0 = \source $,
		\item for each state $ s \in \St $
		\begin{itemize}
			\item $ \Ac_1(s) = \{ \{ \T \cup \{\sink \} \} \times \{0\}^3 \} $, $ \Ac_2(s) = \{ \varepsilon \} $, if $ s = \tuple{\source,\{0\}^3} $,
			\item $ \Ac_1(s) = \{ \term^1_i : s[0] \subseteq \term_i \wedge i \in [1,k] \} $, $ \Ac_2(s) = \{ 0,1 \}^{3}$, if $ s \in \{\T \times \{0\}^3 \} $,
			\item $ \Ac_1(s) = \{ \varepsilon \} $, $ \Ac_2(s) = \{ \varepsilon \} $, if $ s \in \bigcup_{j \in [1,k]} (\term_j \times \{0,1\}^3) $,
		\end{itemize}
		\item for an action profile $ \jact = (\act_{1},\act_{2}) $
		\begin{itemize}
			\item $ \trnFun(s,\jact) = \act_{1} $, if $ s = \tuple{\source,\{0\}^3} $,
			\item $ \trnFun(s,\jact) = \tuple{\act_{1},\act_{2}} $, if $ s \in \{\T \times \{0\}^3 \} $,
			\item $ \trnFun(s,\jact) = \tuple{\term_j^{(i \bmod 3)+1 },s[1]} $, if $ s = \tuple{\term^i_j,s[1]} \in \bigcup_{j \in [1,k]} (\term_j \times \{0,1\}^3) $;
			\item $ \trnFun(s,\jact) = s $, otherwise;
		\end{itemize}
		\item for each state $ s \in \St, \labFun(s) = s[0] $,
		\item for each state $ s \in \St $
		\begin{itemize}
			\item $ \wFun_1(s) = \frac{2}{3} $, if $ s[0] = \sink $\footnote{This can be implemented by a macrostate with three substates---2 substates with weight of 1, and 1 with weight of 0---forming a simple cycle.},
			\item $ \wFun_1(s) = 0 $, otherwise;
		\end{itemize}
		\item the payoff of player $ i \in \Ag $ for an ultimately periodic path $ \pi $ in $ \Game $ is
		\begin{itemize}
			\item $ \pay_1(\pi) = \MP(\wFun_1(\pi)) $,
			\item $ \pay_2(\pi) = - \MP(\wFun_1(\pi)) $,
		\end{itemize}
	\end{itemize}
	Furthermore, let $ \beta = |\setx| $ and the \GRone property to be $ \phi := \always \sometime~\lnot \sink $. Define a (partial) reward scheme $ \kappa : \setx \to \{0,1\} $. The weights are updated with respect to $ \kappa $ as follows:\\	
	for each $ s \in \St $ such that $ s[0] \in \term_j \setminus \sety $, that is, an $ x $-literal that appears in term $ t_j $
	\[
	\wFun_1(s) =
	\begin{cases}
	1, & \text{if } \kappa(s) = 1 \wedge s[0]~\text{is not negated in}~t_j \\
	1, & \text{if } \kappa(s) = 0 \wedge s[0]~\text{is negated in}~t_j \\
	0, & \text{if } \kappa(s) = 1 \wedge s[0]~\text{is negated in}~t_j \\
	0, & \text{otherwise;}
	\end{cases}
	\]
	for each $ s \in \St $ such that $ s[0] \in \term_j \setminus \setx $, that is, a $ y $-literal that appears in term $ t_j $, $ s[0] = \term^i_j $
	\[
	\wFun_1(s) =
	\begin{cases}
	1, & \text{if } s[1][i] = 1 \wedge s[0]~\text{is not negated in}~t_j \\
	1, & \text{if } s[1][i] = 0 \wedge s[0]~\text{is negated in}~t_j \\
	0, & \text{if } s[1][i] = 1 \wedge s[0]~\text{is negated in}~t_j \\
	0, & \text{otherwise;}
	\end{cases}
	\]
	the weights of other states remain unchanged.
	
	The construction is now complete, and polynomial to the size of formula $ \psi(\setx,\sety) $. We claim that $ \SI(\Game,\phi,\beta) \neq \varnothing $ if and only if there is an $ \vecx \in \{0,1\}^n $ such that $ \psi(\setx,\sety) $ is true for every $ \vecy \in \{0,1\}^m $. From left to right, consider a reward scheme $ \kappa \in \SI(\Game,\phi,\beta) $ which implies that there exists no Nash equilibrium run in $ (\Game,\kappa) $ that ends up in $ \sink $. This means that for every action $ \jact_{2} \in \Ac_2(s), $ there exists $ \jact_1 \in \Ac_1(s), s \in \{\T \times \{0\}^3 \} $, such that $ \pay_1(\pi) = 1 $, where $ \pi $ is the resulting path of the joint action. Observe that this corresponds to the existence of (at least) a term $ t_i $, which evaluates to true under assignment $ \vecx $, regardless the value of $ \vecy $. From right to left, consider an assigment $ \vecx \in \{0,1\}^n $ such that for all $ \vecy \in \{0,1\}^m $, the formula $ \psi(\setx,\sety) $ is true. This means that for every $ \vecy $, there exists (at least one) term $ t_i $ in $ \psi(\setx,\sety) $ that evaluates to true. By construction, specifically the weight updating rules, for every $ \jact_{2} $ corresponding to assignment $ \vecy $, there exists $ \term_j $ such that $ \forall i \in [1,3], \wFun_1(\term^i_j) = 1 $. This means that player 1 can always get payoff equals to 1, therefore, any run that ends in $ \sink $ is not sustained by Nash equilibrium.
\end{proof}

\section{Optimality and Uniqueness of Solutions}\label{sec:opt-unique}
Having asked the questions studied in the previous sections, the principal may want to dig deeper. Because the power of the principal is limited by its budget, and because from the point of view of the system, it may be associated with a reward (e.g., money, savings, etc.) or with the inverse of the amount of a finite resource (e.g., time, energy, etc.) an obvious question is asking about {\em optimal} solutions. This leads us to {\em optimisation} variations of the problems we have studied. Informally, in this case, we ask what is the least budget that the principal needs to ensure that the implementation problems have positive solutions. The principal may also want to know whether a given reward scheme is {\em unique}, so that there is no point in looking for any other solutions to the problem. In this section, we investigate these kind of problems, and classify our study into two parts, one corresponding to the \weak problem and another one corresponding to the \strong problem.

\subsection{Optimality and Uniqueness in the Weak Domain}
We can now define formally some of the problems that we will study in the rest of this section. To start, the optimisation variant for \weak is defined as follows.

\begin{definition}[\optwi]
	Given a game $\Game$ and a specification formula $\phi$:
	\begin{center}
		What is the optimum budget $ \beta$ such that $ \WI(\Game,\phi,\beta) \neq \varnothing $?
	\end{center}
\end{definition}

Another natural problem, which is related to \optwi, is the ``exact'' variant -- a membership question. In this case, in addition to $ \Game $ and $ \phi $, we are also given an integer $ b $, and ask whether it is indeed the smallest amount of budget that the principal has to spend for some optimal weak implementation. This decision problem is formally defined as follows.

\begin{definition}[\exactwi]
	Given a game $\Game$, a specification formula $\phi$, and an integer $ b $:
	\begin{center}
		Is $b$ equal to the optimum budget for $ \WI(\Game,\phi,\beta) \neq \varnothing $?
	\end{center}
\end{definition}

To study these problems, it is useful to introduce some concepts first. More specifically, let us introduce the concept of {\em implementation efficiency}. We say that a \weak (resp.\ \strong) is \textit{efficient} if $ \beta = \cost(\kappa) $ and there is no $ \kappa' $ such that $ \cost(\kappa') < \cost(\kappa) $ and $ \kappa' \in \WI(\Game,\phi,\beta) $ (resp.\ $ \kappa' \in \SI(\Game,\phi,\beta) $). In addition to the concept of efficiency for an implementation problem, it is also useful to have the following result.

\begin{proposition}\label{lem:opt-bound}
	Let $ z_i $ be the largest payoff that player $ i $ can get after deviating from a path $\pi$. The optimum budget is an integer between 0 and $ \sum_{i \in \Ag} z_i \cdot (|\St|-1) $.
\end{proposition}

\begin{proof}
	The lower-bound is straightforward. The upper-bound follows from the fact that the maximum value the principal has to pay to player $ i $ is when the path $ \pi $ is a simple cycle and formed from all states in $ \St $, apart from 1 deviation state. 
\end{proof}

Using Proposition~\ref{lem:opt-bound}, we can show that both \optwi and \exactwi can be solved in PSPACE for \LTL specifications. Intuitively, the reason is that we can use the upper bound given by Proposition~\ref{lem:opt-bound} to go through all possible solutions in exponential time, but using only nondeterministic polynomial space. Formally, we have the following results. 

\begin{theorem}\label{thm:optwi-ltl}
	\optwi with \LTL specifications is \fpspace-complete.
\end{theorem}

\begin{proof}
	Since the search space is bounded (Proposition \ref{lem:opt-bound}), by using \weak an an oracle we can iterate through every instance and return the smallest $ \beta $ such that $ \WI(\Game,\phi,\beta) \neq \varnothing $. Moreover, each instance is of polynomial size in the size of the input. Thus membership in \pspace follows. Hardness is straightforward.
\end{proof}

\begin{corollary}\label{cor:exactwi-ltl}
	\exactwi with \LTL specifications is \pspace-complete.
\end{corollary}

The fact that both \optwi and \exactwi with \LTL specifications can be answered in, respectively, \fpspace and \pspace does not come as a big surprise: checking an instance can be done using polynomial space and there are only exponentially many instances to be checked. However, for \optwi and \exactwi with \GRone specifications, these two problems are more interesting. 

\begin{theorem}\label{thm:optwi-gr1}
	\optwi with \GRone specifications is $ \FP^{\np} $-complete.
\end{theorem}

\begin{proof}
	Membership follows from the fact that the search space, which is bounded as in Proposition \ref{lem:opt-bound}, can be fully explored using binary search and \weak as an oracle. More precisely, we can find the smallest budget $ \beta $ such that $ \WI(\Game,\phi,\beta) \neq \varnothing $ by checking every possible value for $ \beta $, which lies between 0 and $ 2^n $, where $ n $ is the length of the encoding of the instance. Since we need logarithmically many calls to the \np oracle (to \weak), in the end we have searching procedure that runs in polynomial time. 
	
	For hardness we reduce from \textsc{TSP Cost} (optimal travelling salesman problem) that is known to be $ \FP^{\np} $-complete \cite{1994-papadimitriou}. Given a \textsc{TSP Cost} instance $ \tuple{G, c} $, $ G = \tuple{V, E} $ is a graph, $ c : E \to \SetZ $ is a cost function. We assume that $ \WI(\Game,\phi,\beta) $ is efficient. To encode \textsc{TSP Cost} instance, we construct a game $ \Game $ and \GRone formula $ \phi $, such that the optimum budget $ \beta $ corresponds to the value of optimum tour. Let $ \Game $ be such a game where
	\begin{itemize}
		\item $ \Ag = \{ 1 \} $,
		\item $ \St = \{ \tuple{v,e} : v \in V \wedge e \in \IN(v) \} \cup \{ \tuple{\sink,\varepsilon} \} $,
		\item $ s_0 $ can be chosen arbitrarily from $ \St \setminus \{\tuple{\sink,\varepsilon}\} $,
		\item for each state $ \tuple{v,e} \in \St $ and edge $ e' \in E \cup \{\varepsilon\} $
		\begin{itemize}
			\item $ \trnFun(\tuple{v,e},e') = \tuple{\trg(e'),e'}, \text{if}~v \neq \sink~\text{and}~e' \neq \varepsilon, $
			\item $ \trnFun(\tuple{v,e},e') = \tuple{\sink,\varepsilon}, \text{otherwise}; $
		\end{itemize}			
		\item for each state $ \tuple{v,e} \in \St $
		\begin{itemize}
			\item $ \wFun_1(\tuple{v,e}) = \max\{c(e'): e' \in E \} - c(e) $, if $ v \neq \sink $,
			\item $ \wFun_1(\tuple{v,e}) = \max\{c(e'): e' \in E \} $, otherwise;
		\end{itemize}
		\item the payoff of player 1 for a path $ \pi $ in $ \Game $ is $ \pay_1(\pi) = \MP[\wFun_1(\pi)] $,
		\item for each state $ \tuple{v,e} \in \St $, the set of actions available to player 1 is $ \OUT(v) \cup \{\varepsilon\} $,
		\item for each state $ \tuple{v,e} \in \St $, $ \labFun(\tuple{v,e}) = v $.
	\end{itemize}
	Furthermore, let $ \phi := \bigwedge_{v \in V} \always \sometime~v $. The construction is now complete, and is polynomial to the size of $ \tuple{G, c} $.
	
	Now, consider the smallest $ \cost(\kappa), \kappa \in \WI(\Game,\phi,\beta) $. We argue that $ \cost(\kappa) $ is indeed the lowest value such that a tour in $ G $ is attainable. Suppose for contradiction, that there exists $ \kappa' $ such that $ \cost(\kappa') < \cost(\kappa) $. Let $ \pi' $ be a path in $ (\Game,\kappa') $ and $ z_1 = \wFun_1(\tuple{\sink,\varepsilon}) $ the largest value player 1 can get by deviating from $ \pi' $. We have $ \pay_1(\pi') < z_1 $, and since for every $ \tuple{v,e} \in \St $ there exists an edge to $ \tuple{\sink,\varepsilon} $, thus player 1 would deviate to $ \tuple{\sink,\varepsilon} $ and stay there forever. This deviation means that $ \phi $ is not satisfied, which is a contradiction to $ \kappa' \in \WI(\Game,\phi,\beta) $. The construction of $ \phi $ also ensures that the path is a valid tour, i.e., the tour visits every city at least once. Notice that $ \phi $ does not guarantee a Hamiltonian cycle. However, removing the condition of visiting each city \textit{only once} does not remove the hardness, since \textit{Euclidean} \textsc{TSP} is \np-hard \cite{GGJ1976,PAPADIMITRIOU1977237}. Therefore, in the planar case there is an optimal tour that visits each city only once, or otherwise, by the triangle inequality, skipping a repeated visit would not increase the cost. Finally, since $ \WI(\Game,\phi,\beta) $ is efficient, we have $ \beta $ to be exactly the value of the optimum tour in the corresponding \textsc{TSP Cost} instance.
\end{proof}

\begin{corollary}\label{cor:exactwi-gr1}
	\exactwi with \GRone specifications is \DP-complete.
\end{corollary}

\begin{proof}
	For membership, observe that an input is a ``yes'' instance of \exactwi if and only if it is a ``yes'' instance of \weak~{\em and} a ``yes'' instance of \weakc (the problem where one asks whether $ \WI(\Game,\phi,\beta) = \varnothing $). Since the former problem is in \np and the latter problem is in co\np, membership in \DP follows. For the lower bound, we use the same reduction technique as in Theorem~\ref{thm:optwi-gr1}, and reduce from \textsc{Exact TSP}, a problem known to be \DP-hard \cite{1994-papadimitriou,PAPADIMITRIOU1984244}.
\end{proof}

Following \cite{Papadimitriou1984}, we may naturally ask whether the optimal solution given by \optwi is unique. We call this problem \uoptwi. For some fixed budget $ \beta $, it may be the case that for two reward schemes $ \kappa, \kappa' \in \WI(\Game,\phi,\beta) $ -- we assume the implementation is efficient -- we have $ \kappa \neq \kappa'$ and $ \cost(\kappa) = \cost(\kappa') $. With \LTL specifications, it is not difficult to see that we can solve \uoptwi in polynomial space. Therefore, we have the following result.

\begin{corollary}\label{cor:uoptwi-ltl}
	\uoptwi with \LTL specifications is \pspace-complete.
\end{corollary}

For \GRone specifications, we reason about \uoptwi using the following procedure:
\begin{enumerate}
	\item Find the exact budget using binary search and \weak as an oracle;
	\item Use an \np oracle once to guess two distinct reward schemes with precisely this budget; if no such reward schemes exist, return ``yes''; otherwise, return ``no''.
\end{enumerate}

The above decision procedure clearly is in $ \DeltaPTwo $ (for the upper bound). Furthermore, since Theorem \ref{thm:optwi-gr1} implies $ \DeltaPTwo $-hardness \cite{KRENTEL1988490} (for the lower bound), we have the following corollary.

\begin{corollary}\label{cor:uoptwi-gr1}
	\uoptwi with \GRone specifications is $ \DeltaPTwo $-complete.
\end{corollary}

\subsection{Optimality and Uniqueness in the Strong Domain}
In this subsection, we study the same problems as in the previous subsection but with respect to the \strong variant of the equilibrium design problem. We first formally define the problems of interest and then present the two first results.

\begin{definition}[\optsi]
	Given a game $\Game$ and a specification formula $\phi$:
	\begin{center}
		What is the optimum budget $ \beta$ such that $ \SI(\Game,\phi,\beta) \neq \varnothing $?
	\end{center}
\end{definition}

\begin{definition}[\exactsi]
	Given a game $\Game$, a specification formula $\phi$, and an integer $ b $:
	\begin{center}
		Is $b$ equal to the optimum budget for $ \SI(\Game,\phi,\beta) \neq \varnothing $?
	\end{center}
\end{definition}

For the same reasons discussed in the weak versions of these two problems, we can prove the following two results with respect to games with \LTL specifications. 

\begin{theorem}\label{thm:optsi-ltl}
	\optsi with \LTL specifications is \fpspace-complete.
\end{theorem}

\begin{proof}
	The proof is analogous to that of Theorem \ref{thm:optwi-ltl}.
\end{proof}

\begin{corollary}\label{cor:exactsi-ltl}
	\exactsi with \LTL specifications is \pspace-complete.
\end{corollary}

For \GRone specifications, observe that using the same arguments for the upper-bound of \optwi with \GRone specifications, we obtain the upper-bound for \optsi with \GRone specifications. Then, it follows that \optsi is in $ \FP^{\SigmaPTwo} $. For hardness, we define an $ \FP^{\SigmaPTwo} $-complete problem, namely $ \MWQSAT_2 $. Recall that in $ \QSAT_2 $ we are given a Boolean 3DNF formula $ \psi(\setx,\sety)$ and sets $\setx =\{ x_1,\dots,x_n \}, \sety = \{ y_1,\dots,y_m \} $, with a set of terms $T =\{ t_1,\dots,t_k \} $. Define $ \MWQSAT_2 $ as follows. Given $ \psi(\setx,\sety) $ and a weight function $ \cFun: \setx \to \mathbb{Z}^{\geq}$, $ \MWQSAT_2 $ is the problem of finding an assignment $ \vecx \in \{0,1\}^n $ with the least total weight such that $ \psi(\setx,\sety) $ is true for every $ \vecy \in \{0,1\}^m $. Observe that $ \MWQSAT_2 $ generalises $ \MQSAT_2 $, which is known to be $ \FP^{\SigmaPTwo[\log n]} $-hard \cite{ChocklerH04}, {\em i.e.}, $ \MQSAT_2 $ is an instance of $ \MWQSAT_2 $, where all weights are~1.

\begin{theorem}
	$ \MWQSAT_2 $ is $ \FP^{\SigmaPTwo} $-complete.
\end{theorem}

\begin{proof}
	Membership follows from the upper-bound of $ \MQSAT_2 $ \cite{ChocklerH04}: since we have an exponentially large input with respect to that of $ \MQSAT_2 $, by using binary search we will need polynomially many calls to the $ \SigmaPTwo $ oracle. Hardness is immediate~\cite{ChocklerH04}.
\end{proof}

Now that we have an $ \FP^{\SigmaPTwo} $-hard problem in hand, we can proceed to determine the complexity class of \optsi with \GRone specifications. 
For the upper bound we one can use arguments analogous to those in Theorem \ref{thm:optwi-gr1}. For the lower bound, one can reduce from $ \MWQSAT_2 $. Formally, we have: 

\begin{theorem}\label{thm:optsi-gr1}
	\optsi with \GRone specifications is $ \FP^{\SigmaPTwo} $-complete.
\end{theorem}

\begin{proof}
	Membership uses arguments analogous to those in Theorem \ref{thm:optwi-gr1}. For hardness, we reduce $ \MWQSAT_2 $ to \optsi using the same techniques used in Theorem \ref{thm:strong-gr1} with few modifications. Given a \\$ \MWQSAT_2 $ instance $ \tuple{\psi(\setx,\sety),\cFun} $, we construct a game $ \Game $ and \GRone formula $ \phi $, such that the optimum budget $ \beta $ corresponds to the value of optimal solution to $ \tuple{\psi(\setx,\sety),\cFun} $. To this end, we may assume that $ \SI(\Game,\phi,\beta) $ is efficient and construct $ \Game $ with exactly the same rules as in Theorem \ref{thm:strong-gr1} except for the following:
	\begin{itemize}
		\item clearly the value of $ \beta $ is unknown,
		\item the initial weight for each state $ s \in \St $
		\begin{itemize}
			\item $ \wFun_1(s) = \frac{2}{3} $, if $ s[0] = \sink $,
			\item
			\[
			\wFun_1(s) =
			\begin{cases}
			- \cFun(s[0]) + 1, & \text{if } s[0] \in \term_j \setminus \sety \wedge s[0] ~\text{is not negated in}~t_j \\
			1, & \text{if } s[0] \in \term_j \setminus \sety \wedge s[0]~\text{is negated in}~t_j; \\
			\end{cases}
			\]
			\item $ \wFun_1(s) = 0 $, otherwise;
		\end{itemize}
		\item given a reward scheme $ \kappa $, we update the weight for each $ s \in \St $ such that $ s[0] \in \term_j \setminus \sety $, that is, an $ x $-literal that appears in term $ t_j $
		\[
		\wFun_1(s) =
		\begin{cases}
		\wFun_1(s) + \kappa(s), & \text{if } s[0]~\text{is not negated in}~t_j \\
		\wFun_1(s), & \text{otherwise;}
		\end{cases}
		\]
	\end{itemize}
	the construction is complete and polynomial to the size of $ \tuple{\psi(\setx,\sety),\cFun} $.
	
	Let $ o $ be the optimal solution to $ \MWQSAT_2 $ given the input $ \tuple{\psi(\setx,\sety),\cFun} $. We claim that $ \beta $ is exactly $ o $. To see this, consider the smallest $ \cost(\kappa)$, $ \kappa \in \SI(\Game,\phi,\beta) $. We argue that this is indeed the least total weight of an assignment $ \vecx $ such that $ \psi(\setx,\sety) $ is true for every $ \vecy $. Assume towards a contradiction that $ \cost(\kappa) < o $. By the construction of $ \wFun_1(\cdot) $, there exists no $ \pi $ such that $ \pay_1(\pi) > \frac{2}{3} $. Therefore, any run $ \pi' $ that ends up in $ \sink $ is sustained by Nash equilibrium, which is a contradiction to $ \kappa \in \SI(\Game,\phi,\beta) $. Now, since $ \SI(\Game,\phi,\beta) $ is efficient, by definition, there exists no $ \kappa' \in \SI(\Game,\phi,\beta) $ such that $ \cost(\kappa') < \cost(\kappa) $. Thus we have $ \beta $ equals to $ o $ as required.
\end{proof}

\begin{corollary}\label{cor:exactsi-gr1}
	\exactsi with \GRone specifications is $ \DPTwo $-complete.
\end{corollary}
\begin{proof}
	Membership follows from the fact that an input is a ``yes'' instance of \exactsi (with \GRone specifications) if and only if it is a ``yes'' instance of \strong~{\em and} a ``yes'' instance of \strongc, the decision problem where we ask $ \SI(\Game,\phi,\beta) = \varnothing $ instead. The lower bound follows from the hardness of \strong and \strongc problems, which immediately implies \DPTwo-hardness \cite[Lemma 3.2]{Aleksandrowicz2017}.
\end{proof}

Furthermore, analogous to \uoptwi, we also have the following corollaries. 

\begin{corollary}\label{cor:uoptsi-ltl}
	\uoptsi with \LTL specifications is \pspace-complete.
\end{corollary}

\begin{corollary}\label{cor:uoptsi-gr1}
	\uoptsi with \GRone specifications is $ \DeltaPThree $-complete.
\end{corollary}

\section{Equilbrium Design with Social Welfare}\label{sec:sw}

Until this point, we have only considered problems that primarily concern the satisfaction of temporal logic specifications. Indeed, this is one of the key differences between equilibrium design and mechanism design. However, a benevolent principal may not only be concerned with the satisfaction of specifications (and optimality of solutions) but also the well-being of agents and the fairness of outcomes. 
Well-being and fairness can influence the equilibria of a game in many ways. For example, if agents are subject to \textit{inequity aversion}~\cite{fehr1999theory} (i.e., they value fairness and are willing to forego some personal gain), then this can affect the strategies they choose and the resulting equilibria of the game. 
To account for this, we extend implementation problems to include \textit{social welfare} measures.

One well-known measure of social welfare is \textit{utilitarian social welfare}, which provides a measure of overall and average benefit for society. However, this measure may be appropriate for certain circumstances but problematic for others. For instance, it may lead to an unfair reward scheme where the principal allocates all the budget to one player, and none to the others. An alternative measure that takes fairness into consideration is \textit{egalitarian social welfare}. It measures the welfare of a society by the well-being of the worst-off individual (i.e., the maximin criterion). A notable argument in defence of this system is Rawls' \textit{veil of ignorance}~\cite{rawls}. From the players' perspective, this notion of welfare may be more relevant for our setting: prior to the application of a reward scheme, the players do not know which scheme will be chosen and implemented since the principal computes and chooses it ``behind the veil of ignorance''\footnote{In \cite{EM2004} a similar argument for egalitarian system in an \textit{artifical} society is presented using an example of fair access agreement of a common resource (satellite) \cite{LVB1999}.}.

Egalitarian welfare may not only be preferable in a philosophical sense, but also useful in a practical way. To see this, let us revisit our previous Example~\ref{ex:grid-mp} and make some changes to the scenario as follows. For each step robot $ \blacksquare $ takes, it spends 2 units of energy and receives a payment of $-2$. This means that the payoff of robot $ \blacksquare $ in Nash equilibria is $ -\frac{1}{2} $, while robot $ \bigcirc $'s remains $ \frac{1}{2} $. If we consider the robots' payoffs as their energy levels, then the Nash equilibria in such a situation are not attainable: robot $ \blacksquare $ does not have enough energy to realise Nash equilibrium runs. In this setting, the egalitarian social welfare concept is more useful since we want to ensure that every robot has enough resources to carry out its task\footnote{Indeed, the relationship between limited resources and their usage by system components has been extensively studied in the form of games in many research papers, and such games are closely related to games with mean-payoff goals (see e.g., \cite{CdAHS2003,CDHR10,ChatterjeeD12,DDGRT2010,VCDHRR15}).}. However, it is important to note that this analysis comes with some caveats: (a) it generally works only for prefix-independent specifications, and (b) it ignores any finite prefixes, i.e., on some finite prefixes, the payoff may (temporarily) be negative, but is positive as the number of steps approaches infinity. These limitations are typically not a concern for prefix-indepent specifications such as \GRone, which is used in this paper.


Formally, we represent the measure of social welfare as a function $ S: \mathbb{R}^{|\Ag|} \to \mathbb{R} $ mapping a tuple of real numbers into a real number, representing a \textit{payoff aggregation}. More specifically, for a strategy profile $ \vec{\sigma} $, the social welfare measure of $ \vec{\sigma} $ is given by $ S(\pay_1(\vec{\sigma}),\dots,\pay_n(\vec{\sigma})) $. With a slight abuse of notation, we write $ S(\vec{\sigma}) $ for the social welfare measure of $ \vec{\sigma} $. The two aforementioned concepts of social welfare thus can be defined as follows.

\begin{definition}
	For a game $ \Game $ and a strategy profile $ \strpElm $, define
	\begin{enumerate}
		\item the \textit{utilitarian social welfare} of $ \strpElm $ w.r.t. $ \Game $ as $ \usw(\strpElm) = \sum_{i \in \Ag} \pay_i(\strpElm) $;
		
		\item the \textit{egalitarian social welfare} of $ \strpElm $ w.r.t. $ \Game $ as $ \esw(\strpElm) = \min_{i \in \Ag} \{ \pay_i(\strpElm) \} $;
		
		
	\end{enumerate}
\end{definition}

We also denote the minimum utilitarian and egalitarian social welfare for a given set $ \NE(\Game) $ by:
\[ \mi_S(\Game) = \min \{ S(\strpElm) : \strpElm \in \NE(\Game) \} \]
where $ S \in \{ \usw, \esw \} $.

We now define the decision problem of \textit{threshold social welfare}. This problem involves deciding whether there is a reward scheme with a social welfare measure of no less than a given value. Formally, the problem is defined as follows.

\begin{definition}[Threshold social welfare] 
	For a given $ \weak $ (resp. \strong) instance $ (\Game,\phi,\beta) $ , a social welfare function $ S $, and a threshold value $ t $, decide whether there is $ \kappa \in \WI(\Game,\phi,\beta) $ (resp. $ \kappa \in \SI(\Game,\phi,\beta) $) such that $ t \leq \mi_S((\Game,\kappa)) $. We write \utwi and \utsi, for threshold problems with utilitarian social welfare function in \weak and \strong domains, respectively. Similarly, we write \etwi and \etsi for threshold problems with egalitarian social welfare function in \weak and \strong, respectively.
\end{definition}

To solve the threshold problems, we utilise the procedures for \weak and \strong presented in previous sections. We first show how to check the threshold problem where the specification is given in \LTL formulae. For the utilitarian social welfare measure, we have the following.

\begin{theorem}\label{thm:threshold-usw-ltl}
	\utwi and \utsi are \pspace-complete for \LTL specifications.
\end{theorem}

\begin{proof}
	We begin with the upper bound for \utwi. We apply a slight modification of the \weak problem with LTL specifications. Consider the arena $ A' = \tuple{\Ag \cup \{n+1\},\Ac,\St,s_0,\trnFun',\lambda} $, with $ \trnFun' $ defined as
	\[ \trnFun'(\act_{1},\dots,\act_{n},\act_{n+1}) = \trnFun(\act_{1},\dots,\act_{n}) \]
	for every $ (\act_{1},\dots,\act_{n},\act_{n+1}) \in \Ac^{|\Ag|+1} $, and the game $ \Game' = \tuple{A',(\wFun_{i})_{i \in \Ag},(\wFun_{n+1})} $ with $ \wFun_{n+1}(s) = \sum_{i \in \Ag}(\wFun_{i}(s)) $ for each $ s \in \St $. Player $ n+1 $ is a dummy player who does not affect the play of the game. Observe that for every strategy profile $ \vec{\sigma} $ in $ \Game' $, it holds that
	\[ \pay_{n+1}'(\vec{\sigma}) = \sum_{i \in \Ag} \pay_i'(\vec{\sigma}) = \sum_{i \in \Ag} \pay_i(\vec{\sigma}_{-(n+1)}) = \usw(\vec{\sigma}_{-(n+1)}). \]
	
	We can adapt the same procedure for solving the \weak problem in order to solve \utwi. To this end, we replace the \LTLlim formula $ \phi_{\WI} $ with $$ \phi_{\WI,\usw} \coloneqq \phi_{\WI} \wedge (\MP(n+1) \geq t). $$ 
	Observe that if the formula above is satisfied in $ (\Game',\kappa)[z] $, then exists $ \kappa \in \WI(\Game,\phi,\beta) $ such that $ t \leq \mi_{\usw}((\Game,k)) $
	Thus, the \pspace upper bound follows.
	
	The \pspace upper bound for \utsi case can be obtained in a similar way by adapting the procedure for solving the \strong problem. The construction is similar to the one for \utwi explained above, and it suffices to replace the \LTLlim formula $ \phi_{\exists} $ with $$ \phi_{\exists,\usw} \coloneqq \phi_{\exists} \wedge (\MP(n+1) \geq t), $$ and $ \phi_{\forall} $ with
	$$ \phi_{\forall,\usw} \coloneqq \phi_{\exists,\usw} \rightarrow \phi. $$
	
	The lower bounds for both cases immediately follow from the lower bound of \weak and \strong problems since we can reduce those problems to \utwi and \utsi, respectively, by fixing $ t $ to be the smallest weight value appeared in the arena.
\end{proof}

For egalitarian social welfare measure, we obtain the following result.

\begin{theorem}\label{thm:threshold-rsw-ltl}
	\etwi and \etsi are \pspace-complete for \LTL specifications.
\end{theorem}

\begin{proof}
	The upper bounds for \etwi and \etsi follow from the adaptation of Algorithms~\ref{alg:weak} and~\ref{alg:strong} used to solve \weak and \strong, respectively. For \etwi, we begin with the fact that for a given \weak instance $ (\Game,\phi,\beta) $ and a threshold $ t $, we can solve \etwi as follows: 	check if there is a $ \kappa $ such that
	
	\begin{enumerate}
		\item $ \kappa \in  \WI(\Game,\phi,\beta) $, and
		\item $ \mi_{\esw}((\Game,\kappa)) \geq t $.
	\end{enumerate}
	
	Notice that (1) is exactly \weak, and thus can be solved by Algorithm~\ref{alg:weak}. For (2), observe that the constraint $ \mi_{\esw}((\Game,\kappa)) \geq t $ can be ``embedded'' already in line~\ref{proc:find-path} of the algorithm by requiring that $ \pay_i(\pi) \geq t, \forall i \in \Ag $. This can be done by replacing the formula $ \phi_{\WI} $ corresponding to line~\ref{proc:find-path} with the following formula
	$$ \phi_{\WI,\esw} \coloneqq \phi_{\WI} \wedge \bigwedge_{i \in \Ag}(\MP(i) \geq t). $$ 
	
	We can also solve \etsi in a similar way to the above. We modify line~\ref{proc:find-path-strong} in Algorithm~\ref{alg:strong} by adding an extra requirement that the path $ \pi $ satisfies $ \pay_i(\pi) \geq t, \forall i \in \Ag $. To this end, we replace $ \phi_{\exists} $ with 
	$$ \phi_{\exists,\esw} \coloneqq \phi_{\exists} \wedge \bigwedge_{i \in \Ag} (\MP(i) \geq t), $$
	and $ \phi_{\forall} $ with
	$$ \phi_{\forall,\esw} \coloneqq \phi_{\exists,\esw} \rightarrow \phi. $$
	Lower bounds follow directly from the hardness of \weak and \strong by the same argument as in Theorem~\ref{thm:threshold-usw-ltl}.
\end{proof}


%

We now address the threshold problems where the specifications are given in the \GRone fragment. For implementations with \LTL specifications, the bottleneck comes from the \LTL model checking problem. As a result, adding an extra constraint for social welfare threshold would not affect the overall complexity. Surprisingly, this is also true for \GRone specifications: adding the constraint does not incur significant computational cost. 

\begin{theorem}\label{thm:threshold-usw-gr1}
	\utwi and \utsi are \np-complete and $ \SigmaPTwo $-complete, respectively, for \GRone specifications.
\end{theorem}

\begin{proof}
	Again, we use a similar construction as in the proof of Theorem~\ref{thm:threshold-usw-ltl} to build $ A' $ and $ \Game' $, and adapt the procedures for solving the \weak and \strong with \GRone specifications. In both domains, we construct the corresponding multi-weighted graph $ W = (V,E,(\wFun_a')_{a \in \Ag}) $ where $ \wFun_{n+1}'(v) = \wFun_{n+1}(s) - t $. The query to this procedure corresponds exactly to the threshold social welfare problem, giving us the upper bounds. Lower bounds can be obtained by setting $ t = \min\{\wFun_{n+1}(s) : s \in \St\} $.
\end{proof}

\begin{theorem}\label{thm:threshold-rsw-gr1}
	\etwi and \etsi are \np-complete and $ \SigmaPTwo $-complete, respectively, for \GRone specifications.
\end{theorem}

\begin{proof}
	To solve \etwi and \etsi, we directly adapt from the procedures for solving the \weak and \strong with \GRone specifications (Theorems~\ref{thm:weak-gr1} and \ref{thm:strong-gr1}). For \etwi, from the game $\Game{[z]}$, we build the underlying graph $\tuple{V, E, (\wFun_{i}')_{i \in \Ag}}$ where $ \wFun_{i}'(v) = \wFun_{i}(s) - (\max\{z_i,t\}) $. Then we define the linear programs $\LP(\psi_l)$ and $\LP(\theta_1,\dots,\theta_n)$ in the same way. Observe that, one of the two linear programs has a solution if and only if there is a path $\pi$ satisfying $\varphi$ such that for every player $i$, $z_i \leq \pay_i(\pi) $ and $t \leq \pay_i(\pi) $. For \etsi, we also employ a similar construction. To obtain the lower bounds, we reduce from \weak and \strong with \GRone specifications. The reduction simply follows from the fact that by fixing $t = \min \{ \wFun_{i}(s) : i \in \Ag, s \in \St \}$, we can encode \weak and \strong with \GRone specifications into their corresponding social threshold problems.
\end{proof}

\section{Conclusions \& Related and Future Work}\label{sec:conc}

\paragraph*{Equilibrium design vs.\ mechanism design -- connections with economic theory}
Although equilibrium design is closely related to mechanism design, as typically studied in game theory~\cite{HR06}, the two are not exactly the same. Two key features in mechanism design are the following. Firstly, in a mechanism design problem, the designer is not given a game structure, but instead is asked to provide one; in that sense, a mechanism design problem is closer to a rational synthesis problem~\cite{FismanKL10,GutierrezHW15}. Secondly, in a mechanism design problem, the designer is only interested in the game's outcome, which is given by the payoffs of the players in the game; however, in equilibrium design, while the designer is interested in the payoffs of the players as these may need to be perturbed by its budget, the designer is also interested -- and in fact primarily interested -- in the satisfaction of a temporal logic goal specification, which the players in the game do not take into consideration when choosing their individual rational choices; in that sense, equilibrium design is closer to rational verification~\cite{GutierrezHW17-aij} than to mechanism design. Thus, equilibrium design is a new computational problem that sits somewhere in the middle between mechanism design and rational verification/synthesis. Technically, in equilibrium design we go beyond rational synthesis and verification through the additional design of reward schemes for incentivising behaviours in a concurrent and multi-agent system, but we do not require such reward schemes to be incentive compatible mechanisms, as in mechanism design theory, since the principal may want to reward only a group of players in the game so that its temporal logic goal is satisfied, while rewarding other players in the game in an unfair way -- thus, leading to a game with a suboptimal social welfare measure. To remedy this issue, we added social welfare constraints in the design of reward schemes, and showed that such additions do not incur extra cost from computational complexity perspective.

\paragraph*{Equilibrium design vs.\ rational verification -- connections with computer science}
Typically, in rational synthesis and verification~\cite{FismanKL10,GutierrezHW15,GutierrezHW17-aij,KupfermanPV16} we want to check whether a property is satisfied on some/every Nash equilibrium computation run of a reactive, concurrent, and multi-agent system. These verification problems are primarily concerned with qualitative properties of a system, while assuming rationality of system components. However, little attention is paid to quantitative properties of the system. This drawback has been recently identified and some work has been done to cope with questions where both qualitative and quantitative concerns are considered~\cite{AKP18,BBFR13,ChatterjeeD12,CDHR10,ChatterjeeHJ05,GMPRW17,GNPW19,VCDHRR15}. Equilibrium design is new and different approach where this is also the case. More specifically, as in a mechanism design problem, through the introduction of an external principal -- the designer in the equilibrium design problem -- we can account for overall qualitative properties of a system (the principal's goal given by an \LTL or a \GRone specification) as well as for quantitative concerns (optimality of solutions constrained by the budget to allocate additional rewards/resources). Our framework also mixes qualitative and quantitative features in a different way: while system components are only interested in maximising a quantitative payoff, the designer is primarily concerned about the satisfaction of a qualitative (logic) property of the system, and only secondarily about doing it in a quantitatively optimal way. 

\paragraph*{Equilibrium design vs.\ repair games and normative systems -- connections with AI}
In recent years, there has been an interest in the analysis of rational outcomes of multi-agent systems modelled as multi-player games. This has been done both with modelling and with verification purposes. In those multi-agent settings, where AI agents can be represented as players in a multi-player game, a focus of interest is on the analysis of (Nash) equilibria in such games~\cite{BouyerBMU15,GutierrezHW17-aij}. However, it is often the case that the existence of Nash equilibria in a multi-player game with temporal logic goals may not be guaranteed~\cite{GutierrezHW15,GutierrezHW17-aij}. For this reason, there has been already some work on the introduction of desirable Nash equilibria in multi-player games~\cite{AlmagorAK15,Perelli19}. This problem has been studied as a repair problem~\cite{AlmagorAK15} in which either the preferences of the players (given by winning conditions) or the actions available in the game are modified; the latter one also being achieved with the use of normative systems~\cite{Perelli19}. In equilibrium design, we do not directly modify the preferences of agents in the system, since we do not alter their goals or choices in the game, but we indirectly influence their rational behaviour by incentivising players to visit, or to avoid, certain states of the overall system. We studied how to do this in an (individually) optimal way with respect to the preferences of the principal in the equilibrium design problem. However, this may not always be possible, for instance, because the principal's temporal logic specification goal is just not achievable, or because of constraints given by its limited budget. 


\paragraph*{Future work}
This paper considers games with deterministic transitions, perfect information, and non-cooperative setting. There are different interesting directions to extend the work done here. For instance, considering games in which players may be able to cooperate with each other and form coalitions \cite{GKW19,SGW21}. It would also be interesting to investigate classes of imperfect information games in which we may still obtain decidability for our problems \cite{GutierrezPW18,DDGRT2010,BerthonMM17,BelardinelliLMR17,AminofMM14}. Another avenue for future research would be to look into games with probabilistic aspects \cite{KNPS2019,GHLNW21}. 

On the other hand, the reward model proposed in this paper has a \textit{memoryless} structure, meaning that the rewards depend only on the current states and not on the history of the play. This restricts the kinds of equilibria that we can design with this model. For instance, we cannot always design equilibria that implement safety specifications, since they are prefix dependent properties. 
One way to address this limitation is to use a model that incorporates the history of play. For example, finite state machines can be used to create a finite-memory reward scheme modelled by a transducer. This model is similar to a player’s strategy, but the output is a vector of integer rewards. However, since the space of possible schemes is unbounded, a different approach is needed for equilibrium design with this model. Additionally, this reward model is related to the concept of \textit{reward machines} in the reinforcement learning framework \cite{icarte2022reward}. Exploring this direction and incorporating probabilistic aspects of games can establish a connection between equilibrium design and multi-agent reinforcement learning. This is an interesting area of research.

Finally, given that the complexity of equilibrium design is much better than that of rational synthesis/verification, we should be able to have efficient implementations, for instance, as an extension of EVE~\cite{GutierrezNPW18}. 

\section*{Acknowledgements}
Perelli was supported by the PNRR MUR project PE0000013-FAIR
and the PRIN 2020 projects PINPOINT. He was also supported by
Sapienza University of Rome under the “Progetti Grandi di Ateneo”
programme, grant RG123188B3F7414A (ASGARD - Autonomous
and Self-Governing Agent-Based Rule Design).
Wooldridge acknowledges the support of UKRI under a Turing AI World Leading Researcher Fellowship (grant EP/W002949/1).

\bibliographystyle{alphaurl}
\bibliography{bibfinal}
\end{document}